\documentclass[final]{IEEEtran}
\usepackage{url}
\usepackage{amssymb,amsmath}
\usepackage[all]{xy}
\usepackage{bm}
\usepackage[lined,boxed,commentsnumbered]{algorithm2e}
\usepackage{subfigure}
\usepackage{graphicx}
\usepackage{color,psfrag}
\usepackage{xspace}
\usepackage{url}
\usepackage{cite,xspace,xy,soul}
\usepackage{dsfont}
\usepackage{cases}

\begin{document}

\title{Energy-Efficient Relaying over Multiple Slots \\ with Causal CSI}

\author{\IEEEauthorblockN{Chin Keong Ho, Peng Hui Tan, and Sumei Sun \today}
\authorblockA{Institute for Infocomm Research, A*STAR,
Singapore 138632.
\\e-mail: \{hock,phtan,sunsm\}@i2r.a-star.edu.sg}%
}
\author{Chin Keong Ho, Peng Hui Tan, and~Sumei Sun%
\thanks{This work was presented in part at IEEE ICC, Ottawa, Canada, June 2012.}
\thanks{C. K. Ho, P. H. Tan and S. Sun are with the Institute for Infocomm Research, A*STAR, 1 Fusionopolis Way, \#21-01 Connexis, Singapore 138632 (e-mail: \{hock, phtan, sunsm\}@i2r.a-star.edu.sg).}
}

\newtheorem{conjecture}{Conjecture}
\newtheorem{remark}{Remark}
\newtheorem{insight}{Insight}
\newtheorem{question}{Question}
\newtheorem{proposition}{Proposition}
\newtheorem{corollary}{Corollary}
\newtheorem{lemma}{Lemma}
\newtheorem{assumption}{Assumption}
\newtheorem{theorem}{Theorem}
\newtheorem{example}{Example}
\newtheorem{property}[theorem]{Property}

\newcommand{\myse}{\IEEEyessubnumber} 
\newcommand{\myses}{\myse\IEEEeqnarraynumspace} 

\newcommand{\set}[1]{\mathcal{#1}}

\newcommand{\bn}{\begin{enumerate}}
\newcommand{\en}{\end{enumerate}}

\newcommand{\bi}{\begin{itemize}}
\newcommand{\ei}{\end{itemize}}

\newcommand{\be}{\begin{IEEEeqnarray}{rCl}}
\newcommand{\ee}{\end{IEEEeqnarray}}

\newcommand{\benl}{\begin{IEEEeqnarray*}}
\newcommand{\eenl}{\end{IEEEeqnarray*}}

\newcommand{\bel}{\begin{IEEEeqnarray}}
\newcommand{\eel}{\end{IEEEeqnarray}}

\newcommand{\ben}{\begin{IEEEeqnarray*}{rCl}}
\newcommand{\een}{\end{IEEEeqnarray*}}

\newcommand{\barr}{\begin{array}}
\newcommand{\earr}{\end{array}}

\newenvironment{definition}[1][Definition:]{\begin{trivlist}
\item[\hskip \labelsep {\it #1}]}{\end{trivlist}}

\newcommand{\ud}{\mathrm{d}} 

\newcommand{\FigSize}{0.6}
\newcommand{\FigSizeSmall}{0.5}

\newcommand{\avesnr} {\bar{\gamma}} 
\newcommand{\snr} {\gamma} 

\newcommand{\re}[1]{(\ref{#1})}

\newcommand{\Pe} {P_{\mathrm {e}}} 

\newcommand{\goodgap}{%
\hspace{\subfigtopskip}%
\hspace{\subfigbottomskip}}

\newcommand{\dhat}[1]{\Hat{\Hat{#1}}} 
\newcommand{\that}[1]{\Hat{\Hat{\Hat{#1}}}} 
\newcommand{\dtilde}[1]{\Tilde{\Tilde{#1}}} 
\newcommand{\ttilde}[1]{\Tilde{\Tilde{\Tilde{#1}}}} 

\newcommand{\trace}[1]{\mathrm{tr}\{#1\}}

\newcommand{\mi}{I}
\newcommand{\reward}{\mathsf{r}}

\maketitle

\newcommand{\aubopt}{\widetilde{\alpha}}
\newcommand{\Jub}{\widetilde{J}}
\newcommand{\fub}{\widetilde{f}}
\newcommand{\pub}{\widetilde{p}}

\newcommand{\source}{\mathsf S}
\newcommand{\relay}{\mathsf R}
\newcommand{\dest}{\mathsf D}

\newcommand{\bits}{B}

\newcommand{\pow}{\mathbf{p}}
\newcommand{\powone}{p^1}
\newcommand{\pows}{p^{\source}}
\newcommand{\powsmin}{q_{\mathrm{min}}}
\newcommand{\powsc}{\pows}
\newcommand{\powr}{p^{\relay}}
\newcommand{\powro}{p^{\relay\star}}
\newcommand{\powso}{p^{\source\star}}
\newcommand{\powtwo}{p} 
\newcommand{\powtwoopt}{p^{\star}} 
\newcommand{\powtwosub}{\widetilde{p}} 
\newcommand{\powtwolow}{p^{\mathrm{l}}} 
\newcommand{\powtwohigh}{p^{\mathrm{h}}} 

\newcommand{\powall}{P}
\newcommand{\powvec}{\mathbf{p}}
\newcommand{\powsopt}{q^{\star}}
\newcommand{\powscopt}{p^{\mathrm s \star}} 
\newcommand{\powropt}{p^{\mathrm r \star}}
\newcommand{\powallopt}{p^{\mathrm{all \star}}}
\newcommand{\chansr}{\gamma^{\mathrm{sr}}}
\newcommand{\chansd}{\gamma^{\mathrm{sd}}}
\newcommand{\chanrd}{\gamma^{\mathrm{rd}}}
\newcommand{\chanoned}{\gamma^{\mathrm{1d}}}
\newcommand{\chanonetwo}{\gamma^{\mathrm{12}}}
\newcommand{\chanstwo}{\gamma^{\mathrm{s2}}}
\newcommand{\chanmax}{\gamma^{\mathrm{d}}} 
\newcommand{\chanall}{{\boldsymbol{\gamma}}}
\newcommand{\csi}{\mathbf{c}}
\newcommand{\state}{\mathbf{s}}
\newcommand{\action}{\mathbf{a}}
\newcommand{\statespace}{\mathcal{S}}
\newcommand{\actionspace}{\mathcal{A}}

\newcommand{\prob}{\mathrm{P}}

\newcommand{\sou}{\mathsf{S}}

\newcommand{\ami}{m}
\newcommand{\amir}{\ami^{\relay}}
\newcommand{\amid}{\ami^{\dest}}
\newcommand{\amiall}{\mathbf{\ami}}

\newcommand{\Bd}{B^{\dest}}
\newcommand{\Br}{B^{\relay}}
\newcommand{\Ball}{\mathbf{b}}

\newcommand{\snrall}{\bm{\snr}}
\newcommand{\snrsd}{\snr^{\source\dest}}
\newcommand{\snrrd}{\snr^{\relay\dest}}
\newcommand{\snrsr}{\snr^{\source\relay}}

\newcommand{\snrmax}{\widetilde{\snr}}
\newcommand{\snrmaxvec}{\widetilde{\bm{\snr}}}

\newcommand{\Rth}{R^{\small{\sf th}}}
\newcommand{\pth}{p^{\text{th}}}

\newcommand{\snrtr}{\snr_{\mathsf{trunc}}}

\newcommand{\snrave}{\bar{\snr}}
\newcommand{\Eone}[1]{\mathrm{E}_1\left(#1 \right)}

\newcommand{\snrhm}{\widehat{\gamma}}

\newcommand{\avesnrsd}{\bar{\snr}^{\source\dest}}
\newcommand{\avesnrrd}{\bar{\snr}^{\relay\dest}}
\newcommand{\avesnrsr}{\bar{\snr}^{\source\relay}}

\newcommand{\ksd}{c^{\source\dest}}
\newcommand{\krd}{c^{\relay\dest}}
\newcommand{\ksr}{c^{\source\relay}}

\newcommand{\Reff}{R_{\text{eff}}}

\begin{abstract}
In many communication scenarios, such as in cellular systems, the energy cost is substantial and should be conserved, yet there is a growing need to support many real-time applications that require timely data delivery.
To model such a scenario, in this paper we consider the problem of minimizing the expected sum energy of delivering a message of a given size from a source to a destination subject to a deadline constraint.
A relay is present and can assist after it has decoded the message.
Causal channel state information (CSI), in the form of present and past SNRs of all links, is available for determining the optimal power allocation for the source and relay.
We obtain the optimal power allocation policy by dynamic programming and explore its structure.
We also obtain conditions for which the minimum expected sum energy is bounded given a general channel distribution. In particular, we show that for Rayleigh and Rician fading channels, relaying is necessary for the minimum expected sum energy to be bounded. This illustrates the fundamental advantage of relaying from the perspective of energy efficient communications when only causal CSI is available.
Numerical results are obtained which show the reduction in the expected sum energy under different communication scenarios.
\end{abstract}

\section{Introduction}

\subsection{Background}
In delay-sensitive applications, such as multimedia streaming, information is considered to be useful only if it is delivered within a deadline.
This communication scenario is modeled in \cite{Fu06, LeeJindal_TransWCOM, LeeJindal_arXiv} as slotted transmissions with {\em deadline}, i.e.,  a fixed number of bits are to be delivered over $K\geq 1$ slots by a source $\source$ to a  destination $\dest$. Prior to every transmission, {\em causal channel state information} (CSI), consisting of only the past and present slots (but not of future slots), is assumed to be known to $\source$.
The problem of minimizing the sum transmission energy, by allocating energy dynamically over $K$ slots, was first formulated as a finite-horizon dynamic program in \cite{Fu06}. The problem was then specialized to the case where the energy-bit relationship is governed by the AWGN channel capacity formula in \cite{LeeJindal_TransWCOM}.
For such problems with availability of only causal CSI, analytical closed-form solutions are typically not available \cite{Fu06, LeeJindal_TransWCOM}.
The optimality of some scheduling policies was proved in some asymptotic regimes in \cite{LeeJindal_arXiv}.
In other related works, different deadlines were considered for packets that arrived separately in time \cite{NeelyWN09}, while a continuous-time framework was considered in  \cite{ZaferModianoTransACM09}.

Advanced relaying techniques, especially cooperative relaying where the source and relay transmit jointly, have been shown to increase the achievable rate \cite{KramerNow}, reduce the error probability \cite{LanemanWornell03}, or minimize energy \cite{AdachiSunHo12,ChenSerbetliYener08,ZhouCui08,RajanTransVT07,LongbiXiaojunShroffTransNW,NeelyISIT07,hoICC10}. The results in \cite{AdachiSunHo12,KramerNow, LanemanWornell03, ChenSerbetliYener08,ZhouCui08} correspond to the single-slot system where $K=1$. 
In \cite{RajanTransVT07,LongbiXiaojunShroffTransNW,NeelyISIT07,hoICC10}, energy minimization was performed over  multiple slots, allowing further gains to be achieved by optimizing over multiple slots.
In \cite{RajanTransVT07,LongbiXiaojunShroffTransNW}, close-to-optimal schemes were obtained subject to a long-term average delay constraint, but this can lead to a large delay for some message bits.
In \cite{NeelyISIT07}, the optimal scheduler was obtained 
with the assumption that previously received packets were not used for joint decoding.
%
Recently in \cite{hoICC10}, we considered the problem of energy minimization subject to a deadline, assuming causal CSI is available based on a decode-and-forward relaying scheme \cite{KramerNow}.
We have considered {\em intra-slot} relaying,  where {\em every} slot is divided into Phase~1 when the relay listens and Phase~2 when the relay can transmit. Although the power allocation and the phase durations were optimized jointly over all slots, channel coding and decoding was performed independently for each slot, i.e., earlier received packets were not used for joint decoding by the receivers.
If all received packets are instead jointly decoded by all receivers, as in \cite{hoVTC12, Draper11,UrgaonkarNeely11}, then we can realize {\em mutual-information accumulation} (MIA). MIA increases the achievable rate and conversely saves transmission power for a given transmission rate.
The effect of MIA is similar to the use of an ARQ scheme that employs incremental redundancy via retransmissions \cite{Caire01}, except that retransmissions are now deliberate and both $\relay$ and $\dest$ perform decoding.
In these works \cite{hoVTC12,Draper11,UrgaonkarNeely11}, {\em full CSI}, consisting of the past, present and future channel states, is assumed to be available for rate and power allocation. Full CSI, however, may not be available in practice.

\subsection{Problem}

\begin{figure}
\centering
\includegraphics[scale=0.85]{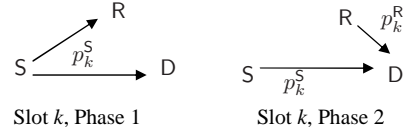}
\label{fig:relaying}
\caption{A two-phase decode-and-forward protocol consisting of $k=1,\cdots,K$ slots. In Phase~1, source transmits with power $\powsc_k$ and relay decodes. In Phase~2, source and relay transmit with power $\powr_k$ and $\pows_k$, respectively.
}
\label{fig:system}
\end{figure}

In this paper, we consider a slotted system where a relay $\relay$ helps source $\source$ to deliver a given number of bits to $\dest$ within a deadline constraint of $K$ slots, based on the two-phase decode-and-forward protocol shown in Fig.~\ref{fig:system}.
Different from the literature, we consider:
(i) only causal CSI is available for power (and hence also rate) allocation, (ii) {\em inter-slot relaying}, i.e., the transition from Phase~1 to Phase~2 occurs only at slot boundaries, and (iii) MIA, i.e., both $\relay$ and $\dest$ use all received packets for joint decoding.
The motivation of using causal CSI and inter-slot relaying is that it is readily implemented in practice, while the motivation of using MIA is to reduce the transmission energy by the use of advanced coding techniques.

We seek to minimize the expected sum energy, by choosing the transmission power of each node dynamically slot by slot.

\subsection{Contributions}

Solving the problem in closed-form is challenging, as the causality of the CSI implies that the present power allocation probabilistically affects the future slots, and the use of MIA creates a further dependence of the power allocation over time.
Nevertheless, we obtain the optimal power allocation policy by dynamic programming, and explore the structural and analytical properties to provide interesting fundamental insights.
Our specific contributions are as follows.
\bi
\item We formulate our problem using a system-level state diagram that depends on whether the relay is active, i.e., if the system is in Phase~1 or 2 as shown in Fig.~\ref{fig:system}.
    This allows us to obtain the optimal power allocation policy via dynamic programming.
\item We obtain conditions for which the minimum expected sum energy is bounded given the channel distribution. In particular, we show that for Rayleigh and Rician fading channels, relaying is {\em necessary} for the minimum expected sum energy to be bounded, i.e., without relaying, the expected sum energy is unbounded even with an optimal power allocation. This points to the fundamental advantage of relaying from the perspective of energy efficient communications when only causal CSI is available.
\item Closed-form results are obtained for specific cases to reveal interesting insights. In particular, our problem of power allocation is related to the problem of deciding at which slot the relay should be activated (by making the relay decode the message). For $K=2$ slots, this problem is solved by minimizing a piecewise convex objective function.
\item Numerical results are obtained for different scenarios to show the potential energy savings, suggesting that most of the savings are recovered with $K=2$.
We propose a heuristic policy with provably bounded expected sum energy; for $K=2$, it performs close to the optimal scheme for small rates.
\ei

This paper is organized as follows. First, Section~\ref{sec:model} gives the system model and formulates the problem of minimizing the expected sum energy.
Section~\ref{sec:causalCSI} solves this problem and analyzes properties of the optimal solution.
Section~\ref{sec:boundedpower} obtains general conditions for which the expected sum energy is bounded.
Numerical results and comparisons are given in Section~\ref{sec:numerical}.
Finally, Section~\ref{sec:con} concludes the paper.

\section{System Model}\label{sec:model}

A source $\source$ delivers a message $W$ of $nB$ nats over $K\in\mathbb{Z}^{+}$ slots to a destination $\dest$ with the help of a relay $\relay$. Each slot consists of $n$ channel uses over time. Thus, the effective transmission rate is $\Reff=B/K$ nats per channel use. Each node has one antenna and is subject to the half-duplex constraint.

\subsection{Coding Scheme}
The source $\sou$ uses independent Gaussian codebooks for each slot $k$.
For every message $W\in\{1,\cdots, e^{nB}\}$,  a codeword $\mathbf{x}^{\mathsf{S}}_k(W)$ of length $n$ is generated at $\sou$ according to $\mathcal{CN}(\mathbf{0}_n,\mathbf{I}_n)$, which denotes the independent and identically distributed (i.i.d.) complex-valued zero-mean unit-variance $n$-variate Gaussian probability density function (PDF).
The relay independently generates its codebook, with codewords $\{\mathbf{x}^{\mathsf{R}}_k(W)\}$, similarly.

\subsection{Relaying Protocol}

\begin{figure}
\centering
\includegraphics[scale=1]{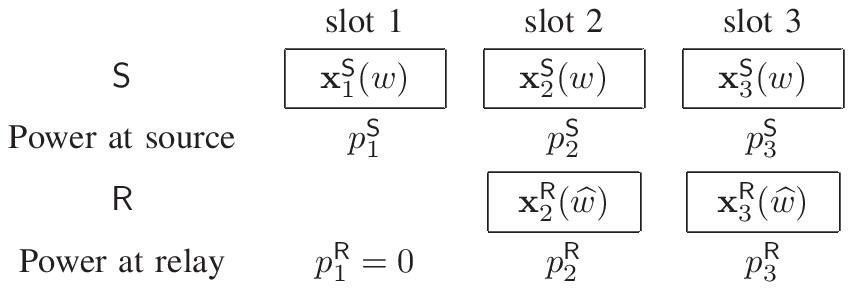}
\caption{Protocol for $K=3$, assuming the relay decodes the message $w$ as $\widehat{w}$ in slot $\widetilde{K}=1$.}
\label{fig_df_fulldup}
\end{figure}

We divide the relaying protocol into two phases, as shown in Fig.~\ref{fig:system}.
We employ a decode-and-forward relaying scheme where each receiver ($\relay$ or $\dest$) can accumulate mutual information over slots.
An example of the power allocation for $K=3$ slots is shown in Fig.~\ref{fig_df_fulldup}.

We assume the wireless communication scenario where the link SNR is constant in each slot, denoted by the positive scalar $\snr$, while the channel phases may vary over $n$ channel uses in each slot, denoted by the length-$n$ vector $\bm{\theta}$.
This assumption is reasonable in practice as communication systems are typically designed such that, over a packet or slot duration, the SNR is constant, while the channel phases may change substantially due to, for example, unavoidable offsets in the carrier frequencies.

In general, for any $\mathsf{Y}$-to-$\mathsf X$ link in slot $k$, let $\snr^{\mathsf{YX}}_k\geq 0$ be the channel SNR and let $\bm{\Theta}_k^{\mathsf{YX}}$ be a diagonal matrix of the $n$ channel phases $\bm{\theta}_k^{\mathsf{YX}}$.
Every receiver~$\mathsf{X}$ has knowledge of both $\snr^{\mathsf{YX}}_k$ and $\bm{\theta}_k^{\mathsf{YX}}$, such as via appropriate training with preambles sent in the packet header.
Since the link SNRs, but not their phases, change very slowly in our assumption, in practice the SNRs can be fed back in time by the receivers.
Hence, we assume that every transmitter~$\mathsf{Y}$ has knowledge of $\snr^{\mathsf{YX}}_k$ but not $\bm{\theta}_k^{\mathsf{YX}}$. This knowledge is used for power allocation in Section~\ref{sec:probform}.

\subsubsection{Phase~1}

Consider Phase~1, before the relay has decoded the message $W$.
The source transmits the codeword with transmission power $\pows_k$.
The codeword is then received at node  $\mathsf{X}\in\{\mathsf{R}, \mathsf{D}\}$, which denotes the relay or the destination respectively, as
\be\label{eqn:phaserx}
\mathbf{y}^{\mathsf{X}}_k = \sqrt{p^{{\mathsf S}}_k \snr^{\mathsf{SX}}_k }  \; \bm{\Theta}_k^{\sou\mathsf{X}}  \; \mathbf{x}^{\mathsf{S}}_k + \mathbf{v}^{\mathsf{X}}_k
\ee
where $\mathbf{v}^{\mathsf{X}}_k\sim\mathcal{CN}(\mathbf{0}_n,\mathbf{I}_n)$ is the additive white Gaussian noise (AWGN).
Using all past received packets $\{\mathbf{y}^{\mathsf{R}}_i, i=1,\cdots,k\}$ for joint typical decoding, the relay reliably decodes message $W$  if \cite{Cover06}
\be\label{eqn:sisodecode}
\sum_{i=1}^{k} \mi(\pows_i \snrsr_i)\geq B
\ee
where $\mi(x)\triangleq \log(1+\snr)$ is the mutual information function for point-to-point Gaussian channels with received SNR $\snr\geq 0$ and $\log$ is the natural logarithm.
We note that $\mi(\cdot)$ is a strictly increasing concave function.
Thus, the mutual information accumulates over slots until it reaches or exceeds $B$, upon which the relay reliably decodes $W$.

\subsubsection{Phase~2}
We denote $\widetilde{K}$ as the earliest possible slot index such that \re{eqn:sisodecode} holds, where $1\leq \widetilde{K}\leq K-1$; if $\widetilde{K}\geq K$, relaying is not possible and so Phase~2 is not activated.
Consider Phase~2 that runs from slot $k=\widetilde{K}+1$ to $k=K$,  in which the relay can perform relaying. 
Both $\source$ and $\relay$ transmit concurrently using their independent codebooks.
The destination thus receives
\be\label{eqn:phaserx}
\mathbf{y}^{\mathsf{D}}_k =  \sqrt{\pows_k \snrsd_k} \; \bm{\Theta}_k^{{\sou\dest}} \; \mathbf{x}^{\mathsf{S}}_k  + \sqrt{\powr_k\snrrd_k}\; \bm{\Theta}_k^{\relay\dest} \; \mathbf{x}^{\mathsf{R}}_k   + \mathbf{v}^{\mathsf{D}}_k
\IEEEeqnarraynumspace
\ee
for $\widetilde{K}+1\leq k\leq K$.
Since typically the channel phases can change significantly, we assume the extreme case of a {\em phase fading channel} \cite{KramerNow}, i.e.,  each element in $\bm{\theta}^{\sou\dest}$ and $\bm{\theta}^{\relay\dest}$ is  i.i.d. with a uniform PDF over $[0,2\pi)$.
Using {\em all} received packets $\{\mathbf{y}^{\mathsf{D}}_k, k=1,\cdots,K\}$ for joint typical decoding, the destination reliably decodes message $W$  if \cite{KramerNow}
\be\label{eqn:relaydecode}
\sum_{k=1}^{\widetilde{K}} \mi(\pows_k\snrsd_k) + \sum_{k=\widetilde{K}+1}^{K} \mi(\pows_k \snrsd_k + \powr_k \snrrd_k)\geq B.
\ee

\begin{remark}
To obtain \eqref{eqn:relaydecode}, we use independent codebooks at $\source$ and $\relay$, and assume a phase fading channel.
Alternatively, suppose we use an identical (Gaussian) codebook at both nodes followed by a distributed space-time block code%
\footnote{In conventional distributed space-time block codes, the power is fixed, but we shall optimize the power given causal CSI. We shall subsequently show that it is optimal that only one node transmits and so there is, in fact, no need to use a distributed space-time block code.}
\cite{LanemanWornell03}. Then, regardless of the PDF of the phases, the SNRs combine as $\pows_k \snrsd_k + \powr_k \snrrd_k$ in the second phase and so we obtain \eqref{eqn:relaydecode} too.
Hence our subsequent result applies equally for both cases.
\end{remark}

\subsection{Problem Statement}\label{sec:probform}

We consider a general problem formulation where for slot $k\in\mathcal{K}\triangleq \{1,2,\cdots,K\}$, some CSI $\state_k$ is available in some CSI space ${\mathcal{S}}_k$.
Let $\pow_k= (\pows_k, \powr_k)$ be the power allocation and $\snrall_k=( \snrsr_k, \snrsd_k, \snrrd_k)$ be the link SNRs in slot $k$.
We denote a {\em power allocation policy} $\pi\in \Pi$ as a mapping of the CSI $\state_k\in{\mathcal{S}}_k$ to the non-negative power allocation $\pow_k$ for all $k\in\mathcal{K}$, where $\Pi$ is the space of all feasible policies,
i.e.,
\be\label{eqn:policy}
\pi=\{\pow_k(\state_k)\geq \mathbf{0}, \forall \state_k\in {\mathcal{S}}_k, k\in \mathcal{K}\}.
\ee

Our problem is to find the optimal power allocation policy $\pi^{\star}$ so as to minimize the expected sum energy, subject to delivering $B$ nats per channel use in $K$ slots based on the described relaying protocol.
We assume an initial CSI $\state_1$ is (arbitrarily) given.
Mathematically, the optimization problem is given by Problem $P0$:
\be\label{p0}\IEEEyessubnumber
\min_{\pi\in \Pi}  &\;\;& \mathbb{E}\left[ \sum_{k=1}^Kp^{\sou}_{k} +  p^{\relay}_{k} \Big | \state_1 \right] \\ \IEEEyessubnumber \label{p0:c1}
\mbox{\; s.t. } && \widetilde{K} = \arg \min \left \{ 1\leq k \leq K : \sum_{i=1}^{k} \mi(\pows_i \snrsr_i)\geq B \right\}
 \IEEEeqnarraynumspace
\\ \IEEEyessubnumber \label{p0:c2}
 &&  \sum_{k=1}^{\widetilde{K}} \mi(\pows_k\snrsd_k) + \sum_{k=\widetilde{K}+1}^{K} \mi(\pows_k \snrsd_k + \powr_k \snrrd_k)\geq B
\ee
where $\mathbb{E}(\cdot)$ is the expectation over all random variables (conditioned on the initial CSI $\state_1$).
The inequalities in \re{p0:c1} and \re{p0:c2} are due to \re{eqn:sisodecode} and \re{eqn:relaydecode}, to ensure that the message is decoded by the relay and the  destination at the end of $\widetilde{K}$ and $K$ slots, respectively.

In the next section, we shall solve Problem~$P0$ for the specific case where causal CSI is available.

\section{Optimal Power Allocation Policy}\label{sec:causalCSI}

We first formulate Problem~$P0$ equivalently as Problem~$P1$ given causal CSI is available in Section~\ref{sec:equiv_prob}. Then we define a system-level state diagram in Section~\ref{sec:state_diag}, which solves Problem~$P1$ by dynamic programming in Section~\ref{sec:DP}. We provide further analysis in Section~\ref{sec:K=2} for $K=2$ slots.

For convenience, we denote a collection of scalars or vectors with running subscripts as $x_i^j=(x_k,x_{k+1},\cdots, x_j)$ where $k=\max\{1,i\}$ if $i\leq j$ and $j\geq 1$; otherwise, we let $x_i^j=\emptyset$.

\subsection{Equivalent Problem}\label{sec:equiv_prob}

In this paper, we assume that at the start of each slot $k\in\mathcal{K}$, causal CSI is available at all nodes, in terms of the present and past link SNRs $\snrall_1^k\triangleq (\snrall_1,\snrall_{2},\cdots, \snrall_k)$.

Before we specialize to the causal CSI case, we make the following simplification in Lemma~\ref{lem:a}.



\begin{lemma}\label{lem:a}
Suppose that the CSI $\state_k$ is available in slot $k\in \mathcal{K}$.
If the present link SNRs $\snrall_k$ is a deterministic function of $\state_k$ for $k\in \mathcal{K}$,
then there is no loss in optimality in solving Problem~$P0$ if only the {\em stronger node}, i.e., either $\source$ or $\relay$ with the larger (or the same) SNR to $\dest$, transmits in each slot in Phase~2.
\end{lemma}
\begin{proof}
See Appendix~\ref{append:proof:a}.
\end{proof}

The causal CSI $\state_k$ explicitly includes the SNR $\snrall_k$. Thus, Lemma~1 trivially applies. This observation us to reduce the number of power allocation variables, as follows.
Let $\Br_k$ be the additional amount of  mutual information required by $\relay$ to reliably decode the message in slot $k$.
\bi
\item If $\Br_k\leq 0$ (i.e., slot~$k$ is in Phase~2), only the stronger node transmits. We denote the transmission power as $p_k\geq 0$, with the corresponding link SNR given by $\snrmax_k \triangleq \max(\snrsd_k, \snrrd_k)$.
\item If $\Br_k> 0$ (i.e., slot~$k$ is in Phase~1), we re-write $\pows_k$ as $p_k$ with the SNR still as $\snrsd_k$.
\ei
Henceforth, there is only one scalar power $p_k$ to be optimized for every slot, and the policy \re{eqn:policy} simplifies as
\be\label{eqn:policy1}
\pi=\{p_k(\state_k)\geq 0, \forall \state_k\in {\mathcal{S}}_k, k\in \mathcal{K}\}.
\ee
Thus Problem~$P0$ is equivalent to Problem~$P1$, given by
\be\label{p1}\IEEEyessubnumber
\hspace{-0.25cm}
\min_{\pi\in \Pi}  &\;\;& \mathbb{E}\left[\sum_{k=1}^K p_k \Big | \state_1 \right] \\ \IEEEyessubnumber \label{p1:c1}
\hspace{-0.25cm}
\mbox{\; s.t. } 
 &&  \sum_{k=1}^{K} \mi\bigg(p_k \Big(\snrsd_k \mathds{1}(\Br_k>0) + \snrmax_k \mathds{1}(\Br_k\leq 0)\Big)\bigg)  \geq B
\IEEEeqnarraynumspace
\ee
where
$\mathds{1}(\cdot)$ is an indicator function that takes the value of $1$ if the condition in the argument is true and  $0$ otherwise.
Note that via power allocation, we also implicitly determine $\widetilde{K}$ as defined in \re{p0:c2}.


\subsection{State Diagram}\label{sec:state_diag}

Similar to $\Br_k$, let $\Bd_k$ be the additional amount of mutual information required by $\dest$ to reliably decode the message in slot $k$. As $k$ increases, both $\Br_k$ and $\Bd_k$ decrease or remain the same as they receive more mutual information.
We assume $\Br_1=\Bd_1=K\Reff$, meaning that no mutual information is available for both nodes initially; it is easy to extend to the general case where each has some initial side information.
From \re{eqn:sisodecode} and \re{eqn:relaydecode}, we have
\be\label{eqn:ami}
\IEEEyessubnumber\label{eqn:amir}
\Br_k&=&
\left\{
  \begin{array}{ll}
    \Br_{k-1} - \mi(p_{k-1} \snrsr_{k-1} ), & \Br_{k-1}>0 \mbox{ (Phase~1)}; \\
    \Br_{k-1}, & \Br_{k-1}\leq 0 \mbox{ (Phase~2)}
  \end{array}
\right.
\\
\IEEEyessubnumber\label{eqn:amid}
\Bd_k&=&
\left\{
  \begin{array}{ll}
    \Bd_{k-1} - \mi(p_{k-1} \snrsd_{k-1} ), & \Br_{k-1}>0 \mbox{ (Phase~1)}; \\
    \Bd_{k-1}-\mi(p_{k-1} \snrmax_{k-1} ), & \Br_{k-1}\leq 0 \mbox{ (Phase~2)}
  \end{array}
\right.
\IEEEeqnarraynumspace
\ee
for $k\geq 2$.
%
We note that \re{eqn:amir} reflects that the relay does not accumulate mutual information in Phase~2 as it has decoded the message, while \re{eqn:amid} reflects that if the relay can start to perform relaying in Phase~2, then the equivalent SNR is given by $\snrmax_{k-1}$ instead of $\snrsd_{k-1}$.

\renewcommand{\state}{\mathbf{s}}

Besides $\bm{\snr}_1^k$, let  $\Ball_k\triangleq (\Br_k, \Bd_k)$ also be available as causal CSI; there is no loss of generality, since $\Ball_k$ can be calculated from
the past power allocation based on the given policy $\pi$.
Thus, the causal CSI for slot $k$ is
\be\label{eqn:state}
\state_k =(\bm{\snr}_1^k, \Ball_k) \in \mathcal{S}_k
\ee
where $\mathcal{S}_k$ is the corresponding causal CSI space such that $\bm{\snr}_1^k$ and $\Ball_k$ are non-negative.
This definition of the (extended) CSI is useful for us to define the state diagram next.


\newcommand{\termstate}{\mathsf{T}}

\begin{figure}
\centering
\includegraphics[scale=1]{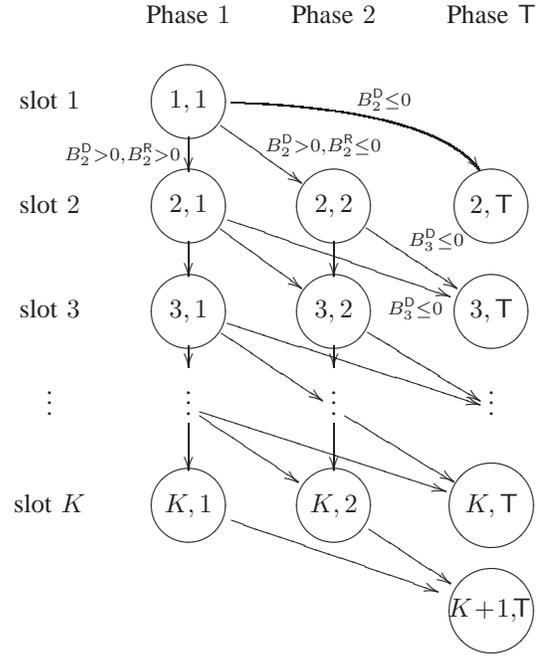}
\caption{All possible transitions of system state $(k,\phi_k)$ for slot $k$ in Phase~$\phi_k$. A transition to $(k+1,\termstate)$ that results in termination occurs if $\Bd_k\leq 0$. Otherwise, a transition to $(k+1,2)$ at Phase~2 occurs at the first occurrence of $\Br_k\leq 0$.}
\label{fig:transitions2} 
\end{figure}

We introduce the termination phase $\termstate$ to indicate that $\Bd\leq 0$, i.e., $\dest$ has decoded the message.
Consider a slot $k\in\mathcal{K}$ in Phase $\phi_k\in\{1,2, \termstate\}$.
We refer to $(k,\phi_k)$ as the {\em system state} and the corresponding causal CSI $\state_{k}$ as its state value%
\footnote{The state value should be denoted as $\state_{k, \phi_k}$ to reflect its full dependence. We maintain the notation $\state_{k}$ for notational convenience, and also because the state value $\state_{k}$  contains $\Br_k$ and $\Bd$ which completely determine the phase $\phi_k$.}.
The transitions of the system states occur as follows, see Fig.~\ref{fig:transitions2}. 
%

\bi
\item Suppose slot $k=1$. The slot must be in Phase~1, hence the system states $(1,2)$ and  $(1,\termstate)$ are empty.
\item Suppose $\Bd_{k+1}\leq 0$. The system state $(k,\phi_k)$ transits to $(k+1,\termstate)$  and the transmission ends.
\item Suppose $\Bd_{k+1}> 0$ and the slot $k\geq 1$ is in Phase~1. The system state $(k,1)$ transits to $(k+1,2)$  if $\Br_{k+1}\leq 0$, as the relay has reliably decoded the message after slot $k$;  otherwise the system state transits to $(k+1, 1)$ and remains in Phase~1.
\item Suppose $\Bd_{k+1}> 0$  and the slot $k>1$ is in Phase~2. The system state $(k,2)$ always transits to $(k+1,2)$ as it stays in Phase~2.
\ei

\begin{remark}\label{rem:markov3}
The state diagram, and hence the subsequent results, apply also to the case where no relaying is performed. This is done by decreasing the link SNR $\snrsr_k$ from the source to relay to zero, in which case Phase~2 will never occur.
\end{remark}


\subsection{Dynamic Programming}\label{sec:DP}

The link SNRs $\snrall_1^K$ are treated in general as random variables with PDF $f_{\snrall_1^K}(\cdot)$.
We treat $p_{k}$ and $\state_k$ also as random variables, since they depend implicitly on the link SNRs via the policy $\pi$.

Theorem~\ref{lem:2} solves Problem~$P1$. We use the Bellman's principle of optimality \cite{Bertsekas}, in which the optimal policy can be obtained by recursively solving a set of so-called Bellman's equation (viz. \eqref{eqn:Jkall}) with the same problem structure (viz. a minimization in \eqref{eqn:Jk}).
It is convenient and intuitive that the indexing of Bellman's equation in \eqref{eqn:Jkall} is two-dimensional over $(k,\phi_k)$, which follows naturally from the  two-dimensional state diagram in Fig.~\ref{fig:transitions2}.
In contrast, if there is no relaying, the indexing is one-dimensional \cite{LeeJindal_TransWCOM,LeeJindal_arXiv}.

We denote the inverse of the mutual information function as $\mi^{-1}(x)\triangleq \exp(x)-1, x \geq 0.$

\begin{theorem}\label{lem:2}
Given causal CSI $\state_1$ at slot~1, the minimum expected sum energy in Problem~$P1$ is given by $J_{1,1}(\state_1)$, which is computed recursively for decreasing $k=K, K-1,\cdots, 1,$ as follows:
\be
\label{eqn:Jkall}
\nonumber
J_{K, \phi_K}(\state_K)
&=& \min_{p_k\geq 0: \Bd_{K+1} \leq 0}  \; p_k \\
\IEEEyessubnumber\label{eqn:JK}
&=& \left \{
\begin{array}{cc}
\frac{\mi^{-1}(\Bd_{K})}{\snrsd_K}, &\phi_K=1; \\ 
\frac{\mi^{-1}(\Bd_{K})}{\snrmax_K}, & \phi_K=2,
\end{array}
\right .
\\
\nonumber
J_{k,\phi_k}(\state_k) &=& \min_{p_k\geq 0}  \; p_k  + \mathbb{E}_{\phi_{k+1}, \state_{k+1}}[J_{k+1,\phi_{k+1}}(\state_{k+1} ) | \phi_k, \state_k],
\\ 
&&\hspace{2.1cm}  \mbox{for } k\in\mathcal{K}\backslash K, \phi_k\in\{1,2\},
\IEEEyessubnumber\label{eqn:Jk}
\\
\IEEEyessubnumber\label{eqn:J_termstate}
J_{k+1,\termstate}(\state_k) &=&0,
\hspace{1.78cm} \mbox{for } k\in\mathcal{K}.
\ee
An optimal policy $\pi^{\star}$ that achieves $J_{1,1}(\state_1)$ is given by the power allocation ${p}^{\star}_k(\state_{k})$ that solves \eqref{eqn:JK} and \eqref{eqn:Jk} for all possible%
\footnote{To reduce computations, we need only consider a state value $\state_{k}$ if it is relevant to its system state $(k,\phi_k)$. For example, if $\phi_k=2$, i.e., the system is in Phase~2,  then $\Br_k$ (related only to the relay) is no longer relevant to future states and can be ignored.}
system states $(k, \phi_k)$ and state values $\state_k$.

\end{theorem}
\begin{proof}
See Appendix~\ref{proof:lem:2}.
\end{proof}

From Theorem~\ref{lem:2}, 
the minimum expected sum energy is given by $J_{1,1}(\state_1)$ that is computed by the Bellman's equation \eqref{eqn:Jkall}, while the power allocation that solves the Bellman's equation form the optimal policy $\pi^{\star}$.
In \re{eqn:Jk}, the first term $p_k$ is the energy used by the current slot, while the second term is the expected  energy accumulated by all future slots $k+1, \cdots, K$.
Thus, besides minimizing the energy used by the current slot, the optimal policy also accounts for the energy used in the future.

The complexity of solving the Bellman's equation via dynamic programming is P-complete, i.e., as hard as any problem with polynomial-time complexity, and so the solution likely cannot be obtained by highly parallel algorithms \cite{Papadimitriou}.
Hence, the (offline) {\em computational complexity} to obtain the optimal power allocation  policy $\pi^{\star}$ is fairly high. However, the (online) {\em implementation complexity} can be made low by first storing
$\pi^{\star}$ in a lookup table, which contains the mapping of all possible state values to the optimal power allocation. During online operation, the table is then used to allocate the power for every slot given the system state and state value.
In practice, the state values are first quantized before the optimal policy is obtained and stored (as is done to obtain numerical results in Section~\ref{sec:numerical}). The exact granularity of the quantization depends on the tradeoff between the complexity/storage of the policy and the expected sum energy.

Next, Theorem~\ref{thm:causalCSI} reveals exactly how $\pi^{\star}$ leads to the optimal transition from Phase 1 to Phase~2.
For full generality, we assume an $m$th-order Markovian channel where the PDF of the link SNRs factorize as
\be\label{eqn:pdfSNR}
f_{\snrall_1^K}(\snrall_1^K)= \prod_{k=1}^K f_{\snrall_k }(\snrall_k| \snrall_{k-m}^{k-1}).
\ee
If $m=K$, this corresponds to the most general case where the PDF of the link SNRs is arbitrary and thus need not be Markovian.
If $m=1$, this corresponds to the first-order Markovian channel, which is widely considered in the literature to reflect the dependence of the present link SNRs only with respect to the most recent past link SNRs, see e.g., \cite{SadeghiKennedyRapajicShams08,HoOostveenLinnartz09}.
If $m=0$, this corresponds to the case of i.i.d. link SNRs.

\begin{theorem}\label{thm:causalCSI}
Assume an $m$th-order Markovian channel where the PDF of the link SNRs is given by  \eqref{eqn:pdfSNR}.
To achieve the minimum expected sum energy in Problem~$P1$,
the Bellman equations \eqref{eqn:JK}, \re{eqn:J_termstate} hold while \eqref{eqn:Jk} can be solved equivalently as
\renewcommand{\state}{\mathbf{s}_k}
\be\label{thm:causalCSI:eqn:0}
\IEEEyessubnumber\label{thm:causalCSI:eqn:1}
J_{k,1}(\state) &=&
\min \{ J_{k,1\rightarrow 2}(\state), J_{k,1\rightarrow 1}(\state)\}
\\
\IEEEyessubnumber
\label{thm:causalCSI:Jk_1b} 
J_{k,1\rightarrow 1}(\state)& \triangleq&
\min_{0\leq R_k< \min\{\Rth_k, \Bd_k\} }  \frac{\mi^{-1}(R_k)}{\snrsd_k} +  J'(R_k)\\
\IEEEyessubnumber
\label{thm:causalCSI:Jk_1c}
J_{k,1\rightarrow 2}(\state) &\triangleq&  \min_{\min\{\Rth_k, \Bd_k\}\leq R_k\leq \Bd_k} \frac{\mi^{-1}(R_k)}{\snrsd_k}  + J''(R_k)
\IEEEeqnarraynumspace
\\
\IEEEyessubnumber\label{eqn:thm:causalCSI:eqn:2}
J_{k,2}(\state)
&=& \min_{0\leq R_k\leq \Bd_k}  \frac{\mi^{-1}(R_k)}{\snrmax_k}  + J'''(R_k)
\ee
where for a given $\state$, we denote
\ben
J'(R_k) &=& \mathbb{E}_{\snrall_{k+1}}[J_{k+1,1}(\snrall_{k+1}, \Ball_k-(g(R_k), R_k)) | \snrall_{k-m+1}^k] \\
J''(R_k)&=&\mathbb{E}_{\snrall_{k+1}}[J_{k+1,2}(\snrall_{k+1}, \Ball_k-(g(R_k), R_k)) | \snrall_{k-m+1}^k]\\
J'''(R_k)&=& \mathbb{E}_{\snrall_{k+1}}[J_{k+1,2}(\snrall_{k+1}, \Ball_k-(0, R_k)) | \snrall_{k-m+1}^k] \\
g(x)&=& \mi\left(\mi^{-1}(x) \snrsr_k /\snrsd_k \right), x\geq 0 \\
\Rth_k&=&\mi\left(\mi^{-1}(\Br_k) \snrsd_k /\snrsr_k \right).
\een
\end{theorem}
\begin{proof}
%
%
First, we prove \re{eqn:thm:causalCSI:eqn:2}.
Suppose the system state is $(k,2)$.  From Fig.~\ref{fig:transitions2}, the system state can transit to $(k+1,2)$ or $(k+1,\termstate)$.
We make a one-to-one transformation from the power variable $p_k$ to the rate variable $R_k=\mi(p_k\snrmax_k)$,  which represents the additional mutual information received by $\dest$ in slot $k$.
The second term in \re{eqn:Jk} expands probabilistically to two cases: the system state transits to $(k+1,2)$ if $0\leq R_k < \Bd_k$ and to $(k+1,\termstate)$ if $R_k \geq  \Bd_k$. To optimally allocate $R_k$ (or $p_k$), we minimizes over both cases.
Thus \re{eqn:Jk} becomes
$J_{k,2}(\state_k)=\min\{Q',Q''\}$
where 
\ben
\label{eqn:thm:proof:1}
Q'&=& \min_{0\leq R_k< \Bd_k}  \frac{\mi^{-1}(R_k)}{\snrmax_k}  +  \mathbb{E}_{\state_{k+1}}\left[J_{k+1,2}(\snrall_{k+1}, \Ball_{k+1}) | \state_k \right], \IEEEeqnarraynumspace\\
Q''&=& \min_{R_k\geq \Bd_k}  \frac{\mi^{-1}(R_k)}{\snrmax_k}  +  \mathbb{E}_{\state_{k+1}}\left[J_{k+1,\termstate}(\snrall_{k+1}, \Ball_{k+1}) | \state_k \right] .
\IEEEeqnarraynumspace\label{eqn:thm:proof:2}
\een
From \re{eqn:J_termstate},  $J_{k+1,\termstate}(\cdot,\cdot)=0$, thus it is optimal to let  $R_k=\Bd_k$ in the minimization problem to obtain $Q''$. 
To obtain $Q'$, we simplify $\Ball_{k+1}=\Ball_k-(0, R_k)$ by using \re{eqn:ami} for Phase~$2$. The only random variable remaining in the expectation in the minimization problem to obtain $Q'$ is $\snrall_{k+1}$, which only depends on $\snrall_{k-m+1}^k$ due to the Markovian property in \re{eqn:pdfSNR}.
By the continuity of both cases at $R_k=\Bd_k$, we obtain \re{eqn:thm:causalCSI:eqn:2}.
%

We sketch the proof for \re{thm:causalCSI:eqn:1}, which is more tedious but follows similarly as before.
Suppose the system state is $(k,1)$.
From Fig.~\ref{fig:transitions2}, the system state can transit to $(k+1,1)$, $(k+1,2)$ or  $(k+1,\termstate)$.
We make a one-to-one transformation from $p_k$ to  $R_k=\mi(p_k\snrsd_k)$.
If $R<\Rth_k=\mi\left(\mi^{-1}(\Br_k) \snrsd_k /\snrsr_k \right)$ and $R< \Bd_k$, the system state transits to  $(k+1,1)$; otherwise the system state transits to  $(k+1,2)$ or  $(k+1,\termstate)$.
Following similar steps, we get \re{thm:causalCSI:eqn:1}, \re{thm:causalCSI:Jk_1b} and \re{thm:causalCSI:Jk_1c}.
\end{proof}

Intuitively, $R_k$ represents the mutual information made available (via power allocation) to the destination in slot~$k$; $\Rth_k$ represents the mutual information threshold for $R_k$ beyond which the system transits from Phase~1 to Phase~2.
In \re{thm:causalCSI:eqn:1}, $J_{k,1\rightarrow 1}(\state)$ and $J_{k,1\rightarrow 2}(\state)$ are the optimal sum energy from slot $k$ to slot $K$, assuming the present slot is in Phase~1, and the next slot is in Phase~1 and Phase~2, respectively. Thus the minimization in \re{thm:causalCSI:eqn:1} is interpreted as a decision making between remaining in Phase~1 or transiting to Phase~2. On the other hand, in \re{eqn:thm:causalCSI:eqn:2} the slot is already in Phase~2, thus such a choice need not be made. Overall, we can interpret the dynamic programming as determining the optimal mutual information or power to allocate, and hence also the optimal slot $\widetilde{K}$ where the transition occurs.

\begin{remark}\label{rem:order}
To solve Problem~$P1$,  it suffices to use %
$\state_k =(\snr_k, \bm{\snr}_{k-m+1}^{k-1}, \Ball_k) \in \mathcal{S}_k$
as a more compact causal CSI, instead of the original causal CSI \re{eqn:state}.
This is because to compute the Bellman equation in Theorem~\ref{thm:causalCSI}, only the compact causal CSI, instead of \re{eqn:state}, is required for all $k=K,\cdots,1$.
Hence, the smaller the Markovian order $m$, the smaller the size of the causal CSI state space $\mathcal{S}_k$.
If $m$ is small, both the computational complexity and implementation complexity of the optimal policy can thus be reduced significantly.
\end{remark}

\subsection{Closed-form Solution for $K=2$ Slots}\label{sec:K=2}

To gain further insight, we consider $K=2$ slots which gives a closed-form solution.

Let $[x]_a^b$ equals $a,x$ or $b$, respectively, if $x<a, a\leq x\leq b$ or $x>b$, and let $[x]^+ = \max(0,x).$
Also, let $\Phi_1 \triangleq \mathbb{E}_{\snrsd_2}\left[ 1/ \snrsd_2 |\snrsd_{1}\right]$ and $\Phi_2 \triangleq \mathbb{E}_{\snrmax_2}\left[ 1/ \snrmax_2  |\snrmax_{1}\right]$.
Clearly, $\Phi_1\geq \Phi_2$.

\begin{corollary}\label{cor:K2}
Consider Problem~$P0$ with $K=2$ 
 and initial CSI $\state=(\snrall, \Ball)$ where $\snrall=(\snrsr_1, \snrsd_1, \snrrd_1)$, $\Ball=(\Bd,\Bd)$.
Then the minimum expected sum energy is given by
\be\label{cor:K2:0}
J_{1,1}(\state)
&=& \min \left\{ \min_{0\leq R< B' }  g_1(R), \min_{B'\leq R\leq \Bd }  g_2(R)) \right\}
\\ \label{cor:K2:1}
&=& \min \{g_1(R_1^{\star}), g_2(R_2^{\star}) \}.
\ee
%
Here, $g_i(R)\triangleq \frac{\mi^{-1}(R)}{\snrsd_1} + \Phi_{i} \mi^{-1}(\Bd-R), i=1,2,$ corresponds to the objective functions in \re{thm:causalCSI:Jk_1b}  and \re{thm:causalCSI:Jk_1c}, respectively; $R_1^{\star}\triangleq [R_1^{\circ}]^+$ and $R_2^{\star}\triangleq [R_2^{\circ}]_{B'}^{\Bd}$ denote the corresponding optimal solutions, where $B'= \min\{\Rth, \Bd\}$,  $\Rth=\mi\left(\mi^{-1}(\Bd) \snrsd_1 /\snrsr_1 \right)$ and $R^{\circ}_i = (\log(\snrsd_{1} \Phi_i) +\Bd)/2, i=1,2$.
\end{corollary}
%
\begin{proof}
Applying \re{eqn:JK}, \re{eqn:J_termstate} and \re{thm:causalCSI:eqn:0} for $K=2$,  we get
%
$J_{1,1}(\state) = \min \{ J_{1,1\rightarrow 2}(\state), J_{1,1\rightarrow 1}(\state)\}
$
where $J_{1,1\rightarrow 1}(\state)= \min_{0\leq R< B' }  g_1(R), J_{1,1\rightarrow 2}(\state) = \min_{B'\leq R\leq \Bd }  g_2(R)$, which proves \re{cor:K2:0}.
Next, we prove \re{cor:K2:1}.
For convenience, a typical graph is shown in Fig.~\ref{fig:f_fun}.


Suppose $\snrsd_1 > \snrsr_1$, thus $\Rth> \Bd$ and $B'=\Bd$. Then trivially $J_{1,1\rightarrow 2}(\state)=g_2(\Bd)$. Now we consider $J_{1,1\rightarrow 1}(\state)$. Since $I^{-1}(\cdot)$ is convex, $g_1(\cdot)$ is a convex function.
The optimal rate $R^{\circ}_1$ for the {\em unconstrained} minimization of $g_1(R)$ is obtained by solving the equation formed by setting the first derivative of $g_1(R)$ to zero, which is given by $R^{\circ}_1=(\log(\snrsd_{1} \Phi_1) +\Bd)/2$.
Since a feasible rate must satisfy $0\leq R< \Bd$, and because  $g_1(R)$ is convex, it follows that the optimal solution that minimizes $g_1(R)$ is given by $R^{\star}_1=[R^{\circ}_1]^+$ or $B'-\epsilon$, where $\epsilon>0$ is infinitesimally small. Without loss of generality, however, we do not need to consider $B'-\epsilon$ as a possible candidate, since it cannot be optimal due to $g_2(B')\leq g_1(B')$.
It can be easily checked that $J_{1,1}(\state)$ is thus given by \re{cor:K2:1} for the case of $\snrsd_1 > \snrsr_1$.
Next, suppose $\snrsd_1\leq \snrsr_1$, thus $\Rth\leq \Bd$ and $B'=\Rth$. Since $g_1(R)$ and $g_2(R)$ are convex functions, it can be shown similarly that the optimal solutions for  $J_{1,1\rightarrow 1}(\state)$ and  $J_{1,1\rightarrow 2}(\state)$ are given by $[R^{\circ}_1]^+$ and $[R^{\circ}_2]_{B'}^{\Bd}$, respectively.
Thus, $J_{1,1}(\state)$ is also given by \re{cor:K2:1}  for the case of $\snrsd_1 \leq \snrsr_1$. This completes the proof.
\end{proof}

We give an intuitive explanation of Corollary~\ref{cor:K2}. Fig.~\ref{fig:f_fun} plots a typical graph of the objective functions $g_1$ and $g_2$ and indicates their respective local optimal solutions $R_1^{\circ}$ and $R_2^{\circ}$, assuming  $\snrsd_1 \leq\snrsr_1$ and so $B'=\Rth$.
In general, we can express \re{cor:K2:0} as
\be\label{eqn:alterK=2}
J_{1,1}(\state)= \min_{0\leq R\leq \Bd} g(R)
\ee
where $g(R)\triangleq g_1(R)\mathds{1}(R\leq B') + g_2(R)\mathds{1}(R\geq B')$.
The objective function $g(R)$ is {\em piecewise convex}, as both $g_1(\cdot)$ and $g_2(\cdot)$ are convex functions.
To obtain the global optimal solution, it is sufficient to check for the boundary and local optimal solutions separately in each region where the convexity of the objective function holds, and then choose the global solution as the one that minimizes $g(R)$. This check gives \re{cor:K2:1}.
For example, given the rate domain $0\leq R\leq B'$, the local optimal solution is $R^{\circ}_1$ and the boundary solutions are $0$ and $B'$. As the objective function $g_1$ is convex, the optimal solution $R_1^{\star}$ in this domain must then be $[R^{\circ}_1]^{+}$.

In general for any $k$, the minimization in Phase~1 is performed over a piecewise non-linear function, so as to decide whether the system state $(k,1)$ should transit to other states.
For $K=2$ slots, this decision can only be made in slot $1$, and moreover 
the objective function 
is a piecewise convex function, thus allowing a closed-form solution to be obtained.
For $K\geq 3$ slots, however, the decision can be made at any slot $k=1,\cdots, K-1$ and the objective function may be a piecewise non-convex function, thus a closed-form solution is difficult to obtain. 
Therefore for $K\geq 3$, we resort to numerical methods in Section~\ref{sec:numerical} based on Theorem~\ref{thm:causalCSI}.
We note that the problem with no relay is already hard to solve, see \cite{LeeJindal_arXiv,LeeJindal_TransWCOM}.

Despite the difficulty in obtaining closed-form solutions in general, we answer the important question of whether the optimal solution leads to bounded expected sum energy in the next section for any $K$.


\begin{figure}
\centering
\includegraphics[scale=0.85]{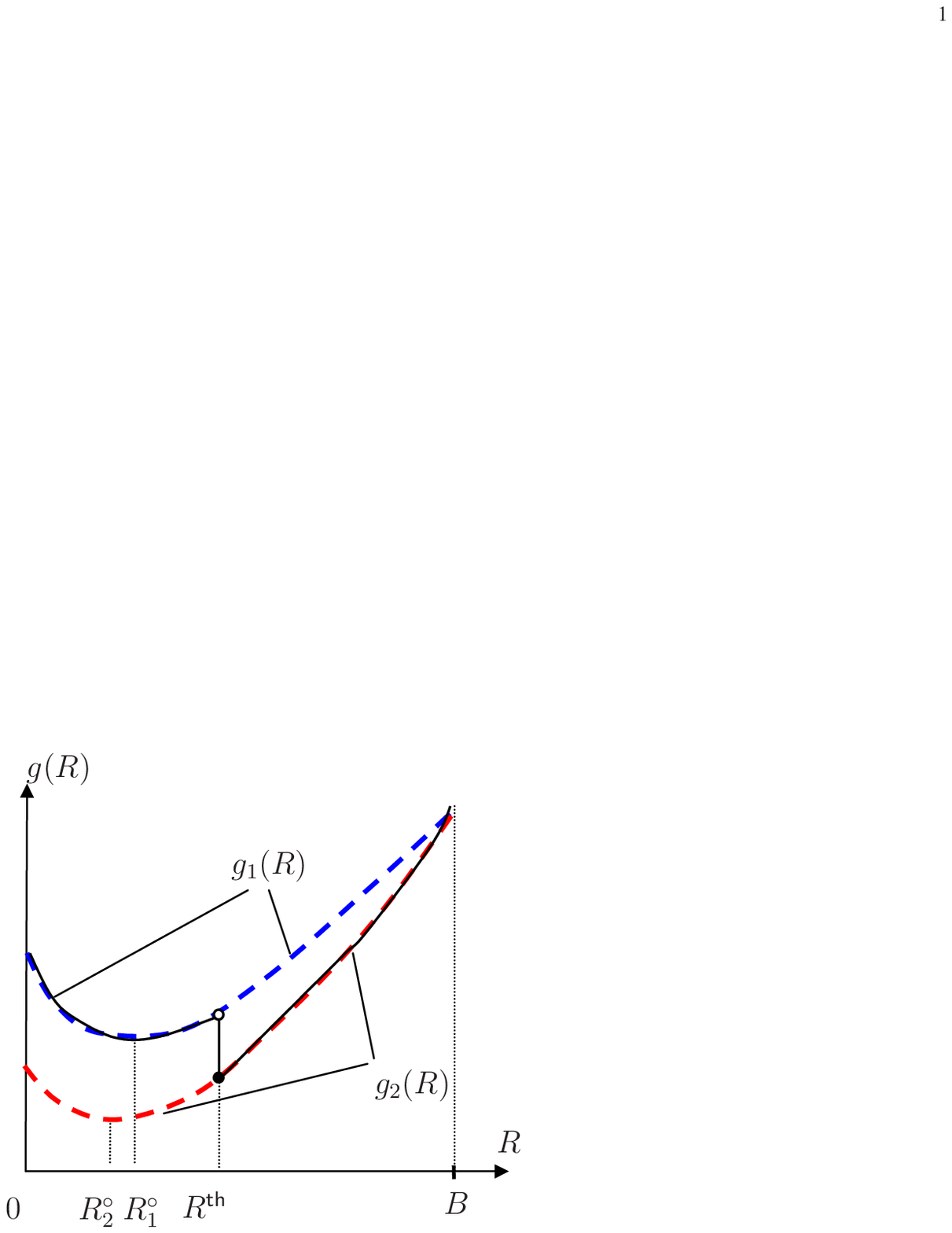}
\label{fig:piecewise_conv}
\caption{Typical graph of objective function $g(R)$. Here, $g(R)=g_1(R)\mathds{1}(R<\Rth) + g_2(R)\mathds{1}(R\geq\Rth)$ with the mutual information threshold $\Rth$, while $R_i^{\circ}$ represents the local optimal $R$ for the objective function $g_i(R)$. }
\label{fig:f_fun}
\end{figure}

%

\section{Conditions for Bounded Expected Sum Energy}\label{sec:boundedpower}


In this section, we consider conditions for the minimum expected sum energy to be bounded for a general channel distribution. In particular, we show that the expected energy is unbounded in Rayleigh and Rician fading channels if relaying is not used.
For simplicity, we assume the link SNRs $\snrsr_k, \snrsd_k, \snrrd_k$ are mutually independent and also independent over slot $k$, but follow the same class of channel distributions.

For exposure, we first assume the link SNRs follows the truncated exponential PDF. Then we consider the general case.

\subsection{Truncated Exponential Distribution}
Consider the {\em truncated exponential PDF} given by
\be\label{eqn:truncexp}
f_{\snr}(\snr; \avesnr)=c\exp(-\snr/\avesnr) 
\ee
for $\snr\geq \snrtr$ and $f_{\snr}(\snr; \avesnr)=0$ otherwise.
The parameter $\snrtr\geq 0$ is a SNR truncation threshold, and 
$c=\exp(\snrtr/\avesnr)/\avesnr$ is a normalization constant.
As $\snrtr\rightarrow 0$, the truncated exponential PDF approaches the exponential PDF with average SNR
$\avesnr$, i.e., the channel amplitude follows a Rayleigh PDF.
For convenience, we shall refer to $\avesnr$ as the average SNR in general, even though it equals the expectation of $\snr$ only if $\snrtr\rightarrow 0$.

We assume all link SNRs follow the  truncated exponential PDF with the same SNR truncation threshold $\snrtr$ but with possibly different corresponding average SNR's $\avesnrsr, \avesnrsd, \avesnrrd$.
These distributions are considered in \cite{LeeJindal_TransWCOM} where no relay is present, with the restriction that $\snrtr$ is strictly positive.
However, in the literature it is typical to assume a Rayleigh-fading channel for wireless communications, which corresponds to $\snrtr=0$.
As such, we also focus on the case where $\snrtr\rightarrow 0$.


For comparing the performance for different $K$, we define the normalized minimum expected sum energy (NMESE) $\bar{J}_K(\Reff)$ as the expected sum energy per slot for transmitting $\Reff$ nats per slot. That is,
\be \label{eqn:expectedJ_K}
\bar{J}_K(\Reff)=\frac{1}{K}\mathbb{E}_{\snrall}\left[ J_{1,1}(\state_1 | K \text{ slots}) \right]
\ee
where $J_{1,1}(\state_1 | K \text{ slots})$ emphasizes that the minimum expected sum energy $J_{1,1}(\state_1)$ is for a $K-$slot system, and the initial system state is $\state_1=(\snrall_1, \Ball_1)$ where $\Ball_1=(K\Reff, K\Reff)$.
For the case of no relaying, we denote its NMESE as $\bar{J}^{\text{no relay}}_K(\Reff)$.
Since relaying requires $K\geq 2$, clearly $\bar{J}_1(\Reff)=\bar{J}^{\text{no relay}}_1(\Reff)$.

For $K=1,2$, we obtain the following respective closed-form solutions (see Appendix~\ref{append:K2trunc})
\be\label{cor:K2trunc:1}
\bar{J}_1(\Reff)
&=& \mathbb{E}\left[\frac{I^{-1}(\Reff)}{\snrsd_1} \right]= \Phi_1 I^{-1}(\Reff) \\
\bar{J}_2(\Reff)
&=& \frac{1}{2} \mathbb{E}\left [ \min \left\{ \frac{\mi^{-1}(R_1^{\star})}{\snrsd_1} +  \Phi_{1} \mi^{-1}(2\Reff-R_1^{\star}),
\right .\right . \nonumber \\
\label{cor:K2trunc:2}
&& \hspace{1.6cm} \left.\left.
\frac{\mi^{-1}(R_2^{\star})}{\snrsd_1} + \Phi_{2} \mi^{-1}(2\Reff-R_2^{\star})\right \} \right ]
\;\;\;\;\;\;
\ee
where $\Phi_1, \Phi_2, R_1^{\star}$ and $R_2^{\star}$ are defined in Section~\ref{sec:K=2}.
Here and henceforth, we assume the expectation is taken over all random variables and drop the subscript of the expectation operator.
We simplify $\Phi_1, \Phi_2$ for our specific channel distribution to give (see Appendix~\ref{append:K2trunc})
\be\label{eqn:1div}
\Phi_1 &=& \mathbb{E} \left[ 1/ \snrsd_1\right]= \ksd \Eone{\snrtr/\avesnrsd} \\ \nonumber
\Phi_2 &=& \mathbb{E} \left[ 1/ \snrmax_1  \right] =  \ksd\Eone{\frac{\snrtr}{\avesnrsd}} + \krd \Eone{\frac{\snrtr}{\avesnrrd}} \\
&& \hspace{1.6cm} -  \hat{k} \Eone{\frac{\snrtr}{\snrhm}}.
\label{eqn:2div}
\ee
Here, $\Eone{x}=\int_x^{\infty}\exp(-t)/t \,\mathrm{d}t$ is the exponential integral, $\snrhm\triangleq {\sf HM}(\avesnrsd_1, \avesnrrd_1)$ is the harmonic mean of $\avesnrsd_1$ and $\avesnrrd_1$, and $\hat{k}\triangleq \exp\left(\snrtr /\snrhm \right)/\snrhm$.
Also, $\ksd=\exp(\snrtr/\avesnrsd)/\avesnrsd, \krd=\exp(\snrtr/\avesnrrd)/\avesnrrd$ are the $c$'s in  \re{eqn:truncexp} corresponding to the respective links.

Consider $\snrtr\rightarrow 0$, i.e., all links approach Rayleigh-fading channels.
Theorem~\ref{thm:boundedenergy} states that the minimum expected sum energy (or the NMESE) is unbounded if there is no relaying for any $K\geq 1$, which is an extremely pessimistic result, but becomes bounded  if relaying is used. 
Note that relaying requires $K\geq2$.




\begin{theorem}\label{thm:boundedenergy}
Suppose that the SNRs for different links and different slots are independent and each link follows the truncated exponential PDF. Then the following holds for the NMESE:
\begin{itemize}
\item[(i)]Without relaying, $\bar{J}^{\text{no relay}}_K(\Reff)\rightarrow\infty$ as $\snrtr\rightarrow 0$ for $K\geq 1$.  
\item[(ii)] With relaying, $\bar{J}_K(\Reff)$ is bounded as $\snrtr\rightarrow 0$ for $K\geq 2$.
\end{itemize}
\end{theorem}
\begin{proof}
We prove the two parts separately. We will need  Lemma~\ref{lem:asym} in Appendix~\ref{append:lem:asym} for part (i).

(i)
For $K=1$, the NMESE with and without relaying is the same. From \re{cor:K2trunc:1}, $\bar{J}^{\text{no relay}}_1(\state)=\bar{J}_1(\Reff)=\Phi_1 I^{-1}(\Reff)$, which approaches infinity because from Lemma~\ref{lem:asym}, $\Phi_1\rightarrow \infty$ as $\snrtr\rightarrow 0$.
Subsequently, we assume $K\geq 2$.
From Remark~\ref{rem:markov3}, the state diagram for the case without relaying can be obtained from Fig.~\ref{fig:transitions2}.
Suppose that the NMESE $\bar{J}^{\text{no relay}}_K(\Reff)$ is {\em bounded} as $\snrtr\rightarrow 0$. 
Then in the optimal policy, we claim that the transition from system state $(K,1)$ to $(K+1,\termstate)$ occurs with zero probability.  Otherwise, if the transition occurs with strictly positive probability, the contribution of the energy to transit to the system state $(K+1,\termstate)$ is $\mi^{-1}(\Bd_{K})/\snrsd_k$. Since the SNRs are independent over slots, the expected energy incurred is $\Phi_1 \mathbb{E}[\mi^{-1}(\Bd_{K})]$, and it goes to infinity as $\snrtr\rightarrow 0$. Thus, the transition to system state $(K+1,\termstate)$ cannot occur.
This implies that the transition from $(K-1,1)$ to $(K,1)$ cannot occur as well (see Fig.~\ref{fig:transitions2} assuming there is no system state in Phase~2).
The state diagram for this $K$-slot system thus reduces to the  $(K-1)$-slot system. By repeating the same arguments that the state transition from $(k,1)$ to $(k+1,\termstate)$ cannot occur for decreasing $k=K-1,\cdots,1$, we are left with the state diagram of a one-slot system. But we have shown that $\bar{J}^{\text{no relay}}_1(\state)\rightarrow \infty$ as $\snrtr\rightarrow 0$.
By contradiction, $\bar{J}^{\text{no relay}}_K(\Reff)$ is {\em unbounded} as $\snrtr\rightarrow 0$.

(ii) It is sufficient to show that a power allocation policy that achieves a bounded expected sum  energy exists.
%
First, consider $K=2$.
We use the following {\em heuristic power allocation policy}. $\source$ transmits in the first slot with power $p_1=\frac{\mi^{-1}(\Reff)}{ \max\{\snrsd_1, \snrsr_1\}}$.
If $\snrsd_1\geq \snrsr_1$, then $\dest$ can decode the message and so the transmission terminates.
If $\snrsd_1< \snrsr_1$, then $\relay$ can decode the message while $\dest$ accumulates mutual information of amount  $\Rth=\mi\left(\mi^{-1}(\Reff) \snrsd_1 /\snrsr_1 \right)$. In the second slot, the stronger node of $\source$ or $\relay$ transmits with power $p_2=\frac{\mi^{-1}(\Reff-\Rth)}{\max\{\snrsd_2, \snrrd_2\}}$, such that $\dest$ reliably decodes the message.
The expected sum energy of this policy is then given by
\be \nonumber 
&& \mathbb{E}\left[p_1 \right]+
\mathbb{E}\left[p_2  \big| \snrsd_1< \snrsr_1\right] \\
&\leq& \mathbb{E}\left[\frac{\mi^{-1}(\Reff)}{ \max\{\snrsd_1, \snrsr_1\}}\right]+
\mathbb{E}\left[\frac{\mi^{-1}(\Reff) }{\max\{\snrsd_2, \snrrd_2\}} \Big| \snrsd_1< \snrsr_1 \right] \nonumber \\
\nonumber 
\\
\nonumber 
&=&\mi^{-1}(\Reff) \left( \mathbb{E}\left[\frac{1}{\max\{\snrsd_1, \snrrd_1\}}\right] +  \mathbb{E}\left[\frac{1}{\max\{\snrsd_2, \snrrd_2\}}\right]  \right)
\\
\label{eqn:proof:thm:boundedenergy:4}
&=&2\mi^{-1}(\Reff) \Phi_2.
\ee
Here, the inequality is because $\Rth\geq 0$ and $I^{-1}(\cdot)$ is an increasing function,
the first equality is due to the independence of the SNR, and the second equality is because the average SNRs are the same over slots.
By Lemma~\ref{lem:asym} in Appendix~\ref{append:lem:asym},  $\Phi_2$ is bounded as $\snrtr\rightarrow 0$.
As the optimal policy achieves at most the same energy as \re{eqn:proof:thm:boundedenergy:4}, $J_{1,1}(\state) $ and its expectation $\bar{J}_2(\Reff)$ are also bounded.

Next, consider $K\geq 3$. For the first $(K-2)$ slots, we can use any power allocation policy that gives a bounded expected sum energy, e.g., $\source$ transmits at some fixed power. For the last two slots, we employ the above heuristic power allocation policy for $K=2$ to deliver the remaining mutual information to $\dest$. This policy, and hence the optimal policy, achieves a bounded expected sum energy for $K\geq 3$.
\end{proof}

Theorem~\ref{thm:boundedenergy} gives the fundamental advantage of using relaying compared to no relaying, from the perspective of minimizing energy.
If no relaying is used, we must therefore accept some probability that the destination cannot be served in Rayleigh-fading channels in practice.
If relaying is used, however, the destination can always be served.

Consider the case of no relaying.
Intuitively, if $K=1$, the expected sum energy is expected to be unbounded due to the lack of channel diversity.
If $K\geq 2$, Theorem~\ref{thm:boundedenergy} states that the NMESE is unbounded, even with an optimal policy.
This may be surprising because under the assumption that {\em full CSI} is available, the NMESE is in fact bounded%
\footnote{We give a sketch of the proof. We use the suboptimal power allocation policy that allocates all the energy to the slot with the largest SNR to deliver all bits to the destination. It can be shown, say by the use of ordered statistics, that the expected energy is bounded. Hence, the optimal policy must also achieve a bounded expected energy.}.
Hence, it is not obvious if the NMESE is bounded given only {\em causal CSI}.
Given only casual CSI, the channels in the future are not known in advance, and based on Theorem~\ref{thm:boundedenergy}, it turns out that it is difficult to effectively exploit the time diversity of the channels. In fact, given simply a one-slot-look-ahead CSI, i.e., assuming the CSI of the present and the next slot is known, is sufficient for the NMESE to be bounded\footnote{To show this, we use the (suboptimal) power allocation scheme that allocates no energy for the first $(K-2)$ slots. Then we deliver all information in the last two slots. In the last two slots, full CSI is available, so we can achieve bounded NMESE.}.
This reinforces the fact that causality plays an important role in the boundedness of the NMESE.

On the other hand, relaying can exploit spatial diversity over different nodes to give a bounded expected energy. This is because the SNRs of {\em different} links are available in the present slot, which can be used effectively.
For example, consider the heuristic power allocation policy proposed in the proof of Theorem~\ref{thm:boundedenergy}, which gives provably bounded expected sum energy.
This power allocation policy exploits spatial diversity in {\em both slots} to give $p_1=\frac{\mi^{-1}(\Reff)}{ \max\{\snrsd_1, \snrsr_1\}}$ and $p_2=\frac{\mi^{-1}(\Reff-\Rth)}{\max\{\snrsd_2, \snrrd_2\}}$.
In contrast, consider the following naive power allocation scheme with relaying over $K=2$ slots: $\source$ transmits such that $\relay$ decodes the message in the first slot, then $\relay$ transmits in the second slot.
Although the advantage of cooperative relaying (via spatial diversity) can be exploited in the second slot, the expected energy in the first slot is unbounded as $\snrtr\rightarrow 0$ (corresponding to a one-slot system).
Thus, exploiting the spatial diversity over {\em only one slot} appears insufficient in a relaying system.

So far, we restrict the study on the boundedness of the NMESE for the Rayleigh distribution, by using the truncated exponential distribution and letting $\snrtr\rightarrow 0$.
In the next section, we generalize the proof technique used in this section to arbitrary channel distributions.

\subsection{Arbitrary Channel Distributions}

We consider an arbitrary channel distribution with the same assumption that channels over different links and slots are independent. Theorem~\ref{thm:genbound} gives conditions to determine if the NMESE is bounded.
Without loss in generality, we assume that all average SNRs are the same, i.e., the channels are i.i.d. over links and slots.

\begin{theorem}\label{thm:genbound}
Suppose that the SNRs for different links and slots are i.i.d. with 
cumulative density function (CDF) $F(\snr), \snr\geq0$.
Then the following holds for the NMESE:
\bi
\item[(i)] Without relaying: $\bar{J}^{\text{no relay}}_K(\Reff)$ is bounded if and only if  $\Phi_1=\mathbb{E}[1/\snr]$ is bounded.
\item[(ii)] With relaying: $\bar{J}_K(\Reff)$ is bounded if $\Phi_2=\mathbb{E}\left[ {1/\max(\snrsd, \snrrd)} \right]$ is bounded.
\ei
\end{theorem}
\begin{proof}
(i) ``Only if" part: 
Following part (i) of the proof of Theorem~\ref{thm:boundedenergy}, it can be easily verified that if $\Phi_1$ is unbounded in general for any CDF, then $\bar{J}^{\text{no relay}}_K(\Reff)$ is unbounded. \\
``If" part:
We allocate power such that $\dest$ reliably decodes the message in the first slot. This incurs an expected energy of $\Phi_1 I^{-1}(\Reff)$, which is bounded if $\Phi_1$ is bounded. Thus $\bar{J}^{\text{no relay}}_K(\Reff)$ is also bounded. 

(ii) 
Suppose we use the heuristic policy in part (ii) of the proof of Theorem~\ref{thm:boundedenergy}. From \eqref{eqn:proof:thm:boundedenergy:4}, the expected sum energy is upper bounded by
$2 \mi^{-1}(\Reff) \Phi_2$. If $\Phi_2$ is bounded, 
$\bar{J}_K(\Reff)$ must also be bounded.
\end{proof}

Theorem~\ref{thm:genbound} clarifies that $\Phi_1$ and $\Phi_2$ that depend only on the channel distribution are the key parameters for checking if the NMESE is bounded.
%
%
%
%
%
%
As an application, we assume that each link SNR $\snr$ follows the noncentral chi-squared PDF with degree of freedom $v>0$ and non-centrality parameter $\lambda>0$.
Without loss of generality for determining whether the NMESE is bounded, let the mean be $\mathbb{E}[\snr]=v$.
Thus, the CDF of $\snr$ is \cite{simon04} 
\be\label{eqn:chi2}
F_{\chi^2}(\snr; v, \lambda)=1- Q_{v/2}\left(\sqrt{\lambda},\sqrt{\snr}\right), \; \snr\geq 0
\ee
where $Q_{M}(a,b)=\int_{b}^{\infty} x (x/a)^{M-1} e^{-(x^2+a^2)/2} I_{M-1}(ax) \,\mathrm{d} x$ is the generalized Marcum Q-function and $I_{M-1}(\cdot)$ is the modified Bessel function of order $M-1$.
Note that we can obtain the noncentral chi-squared distribution via $\snr=\sum_{i=1}^{s}|X_i|^2$ where $X_1,\cdots,X_s$ are complex independent Gaussian variables (not necessarily of zero mean or of the same variance), where $s=v/2$ is a positive integer \cite{simon04}.
We can view this as the SNR distribution after maximal-ratio combining (MRC) of the received signals from $s$ independent Rayleigh or Rician fading channels \cite{simon04}.
This distribution reduces to the Rayleigh fading channel if $v=2$ and $\lambda\rightarrow 0$, and to the Rician fading channel if $v=2$ and $\lambda>0$.


\begin{corollary}\label{thm:boundedenergy_gen}
Suppose that the SNRs for different links and different slots are i.i.d. and each follows the CDF $F_{\chi^2}(\snr; v, \lambda), \snr\geq 0$. Then the following hold for the NMESE: 
\begin{itemize}
\item[(i)] Rayleigh fading channel: $\bar{J}^{\text{no relay}}_K(\Reff)\rightarrow \infty$ for $K\geq1$ and $\bar{J}_K(\Reff)< \infty$ for $K\geq2$.
\item[(ii)] Rician fading channel: $\bar{J}^{\text{no relay}}_K(\Reff)\rightarrow \infty$ for $K\geq1$ and $\bar{J}_K(\Reff)< \infty$ for $K\geq2$.
\item[(iii)] Noncentral chi-squared channel where $v=2s$ and $s\in\{2,3,\cdots\}$: $\bar{J}^{\text{no relay}}_K(\Reff)< \infty$ for $K\geq1$ and $\bar{J}_K(\Reff)< \infty$ for $K\geq2$.
\end{itemize}
\end{corollary}


\begin{proof} 
We prove the three parts separately.

(i) The PDF of the Rayleigh fading channel is given by the exponential PDF $f_{\text{EXP}}(\snr)=e^{-\snr/2}/2$. This PDF is also obtained by the truncated exponential PDF with $\snrtr\rightarrow 0$. From Theorem~\ref{thm:boundedenergy}, we have already obtained $\bar{J}^{\text{no relay}}_K(\Reff)\rightarrow \infty$ for $K\geq1$ and $\bar{J}_K(\Reff)< \infty$ for $K\geq2$.

(ii) First, consider the case of no relaying. The PDF of the Rician fading channel is given by $f_{\chi^2}(\snr; 2, \lambda)= e^{-(\snr+\lambda)/2} I_0(\sqrt{\lambda \snr})/2$.
Since $I_0(\cdot)$ is an increasing function and $I_0(0)=1$ \cite{Abramowitz:1965:HMF},  $ e^{-\lambda/2} f_{\text{EXP}}(\snr)\leq f_{\chi^2}(\snr; 2, \lambda)$ for $\snr \geq 0$.
Thus $ e^{-\lambda/2} \Phi_1^{\text{EXP}} \leq \Phi_1^{\chi^2}$, where  $\Phi_1^{\text{EXP}}$ and $\Phi_1^{\chi^2}$  denote $\mathbb{E}[1/\snr]$ under the exponential and chi-squared PDFs, respectively.
From part~(i), the NMESE is unbounded for Rayleigh fading channels. By part (i) of Theorem~\ref{thm:genbound} (the ``only if" part), $\Phi_1^{\text{EXP}}\rightarrow\infty$, and hence $\Phi_1^{\chi^2}\rightarrow\infty$ too. Using part (i) of Theorem~\ref{thm:genbound} (the ``if" part), we conclude that the NMESE is unbounded for Rician fading channels.

Next, we sketch the proof for the case of relaying. 
We associate the SNR of a Rician fading channel with the SNR of a Rayleigh fading channel with the same non line-of-sight component. Clearly, $\Phi_2$ of the first channel is always lower bounded by $\Phi_2$ of the second channel. From part (ii) of Theorem~\ref{thm:boundedenergy} (the ``only if" part), $\Phi_2$ of the Rayleigh fading channel is always bounded. Thus, $\Phi_2$ of the Rician fading  channel is also bounded. By part (ii) of Theorem~\ref{thm:genbound}, the NMESE with relaying is bounded.

(iii) First, consider the case of no relaying. Suppose $\lambda=0$. Then the PDF of the SNR is a {\em central} chi-squared distribution.  
By direct integration, we obtain $ \mathbb{E}[1/\snr|\lambda=0] = \frac{1}{2 (s-1)} \triangleq \Phi_1^{\text{central}}<\infty$.
Thus, the NMESE is bounded without relaying.
Now suppose $\lambda>0$.
From \cite[Theorem 1]{5429099}, the generalized Marcum Q-function $Q_{M}(a,b)$ is strictly increasing  in $M$ and $a$ for all $a\geq 0$ and $M,b>0$.
It follows from \eqref{eqn:chi2} that $F_{\snr}(\snr; v, \lambda)\leq F_{\snr}(\snr; v, 0)$. We say that the random variable with $\lambda>0$ (call it $\snr_1$) {\em first-order stochastically dominates} the random variable with $\lambda=0$ (call it $\snr_2$).
From decision theory \cite{HadarRussell69}, it is known that $\mathbb{E}[u(\snr_1)] \geq \mathbb{E}[u(\snr_2)]$
for any increasing (utility) function $u$ if $\snr_1$ first-order stochastically dominates $\snr_2$.
Substituting $u(\snr)=-1/\snr$ then gives $\mathbb{E}[1/\snr_1|\lambda>0] \leq \Phi_1^{\text{central}}$. Since we have shown that $\Phi_1^{\text{central}}<\infty$, it follows that $\mathbb{E}[1/\snr_1|\lambda>0]<\infty$.
Finally, with relaying the NMESE must also be bounded (by explicitly allocating no power to the relay as a suboptimal policy).
This completes the proof.
\end{proof}

Part (iii) of Corollary~\ref{thm:boundedenergy_gen} shows that if more than one independent channels are present for SNR combining, the inherent diversity present is sufficient for the NMESE to be bounded, even without relaying. Intuitively, parts (i) and (ii) Corollary~\ref{thm:boundedenergy_gen} suggest that Rayleigh and Rician channels do not have sufficient diversity and relaying is still necessary. Thus, even a strong-line-of-sight signal component (in Rician channels) is not sufficient for energy-efficient communications and so some form of diversity technique is still useful.

\section{Numerical Results}\label{sec:numerical} 
Closed-form solutions cannot be easily obtained for $K\geq3$ in general, as explained in Section~\ref{sec:K=2}. 
Nevertheless, numerical computation of the minimum sum energy can be obtained using Theorem~\ref{thm:causalCSI} by discretizing the state space of the causal CSI $\state_k=(\snrall_k, \Br_k, \Bd_k)$ for all system states.
To this end, we discretize $\Br_k$ and $\Bd_k$ to take values in the set $\mathcal{B}=\{0 \mbox{ to } \Reff  \mbox{ in discrete steps of } \Delta\}$.
Next, we employ a Monte Carlo approach to compute the expectations in the Bellman equations; the computational complexity of this approach is small compared to the alternative approach of discretization of the link SNRs $\snrall_k$ and their probability spaces.
That is, we obtain $N_{\text{sim}}$ independent realizations of $\snrall_1, \cdots, \snrall_K$ with distribution \re{eqn:pdfSNR} where each link SNR follows the PDF \eqref{eqn:truncexp}.
Then we perform a backward recursion of the Bellman's equation with $k$ initialized as $K$.
For our numerical results, we set $\Delta=0.01$ and $N_{\text{sim}}=5000$.
Specifically, we perform the following steps:
\bi
\item Use \re{thm:causalCSI:eqn:0} to compute $J_{k,\phi_k}(\state_k)$ for each $\phi_k\in\{1,2\}$, $\Br_k, \Bd_k \in \mathcal{B}$, and for each realization of $\snrall_k$. If $k=K$, we can use \re{eqn:JK} immediately.
\item Obtain an approximation of the expectation $\mathbb{E}_{\snrall_k}\left [J_{k,,\phi_k}(\state) \right] $ by averaging over all $\snrall_k$.
\item Repeat the first two steps with $k$ decremented by one.
\ei

\begin{figure}
\centering
\includegraphics[scale=\FigSize]{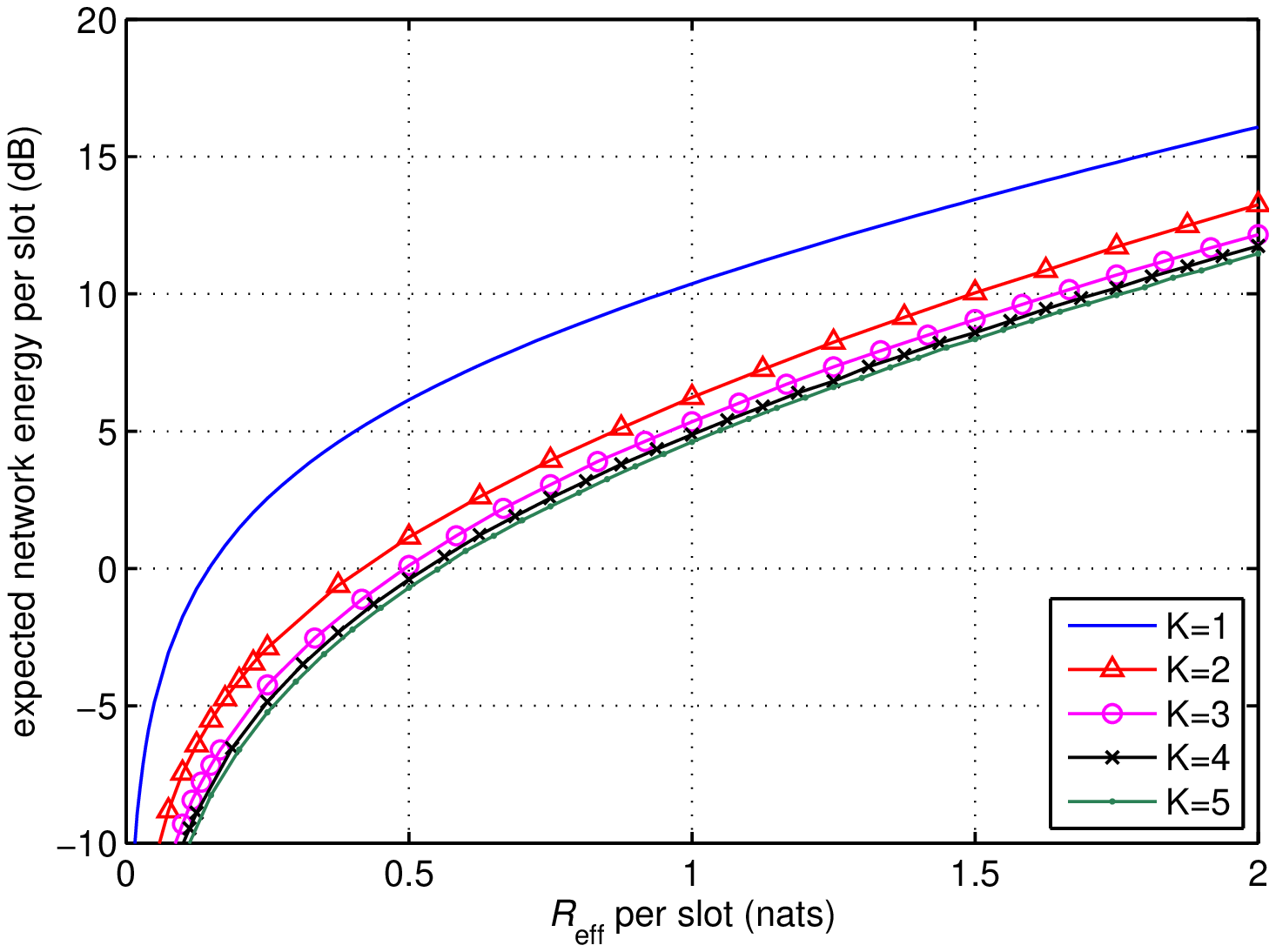}
\caption{Normalized minimum expected sum energy to deliver $\Reff$ nats per slot per channel use over $K$ slots. Relaying and $\snrtr=10^{-3}$. }
\label{fig:smalltrunc_v0}
\end{figure}

\begin{figure}
\centering
\includegraphics[scale=\FigSize]{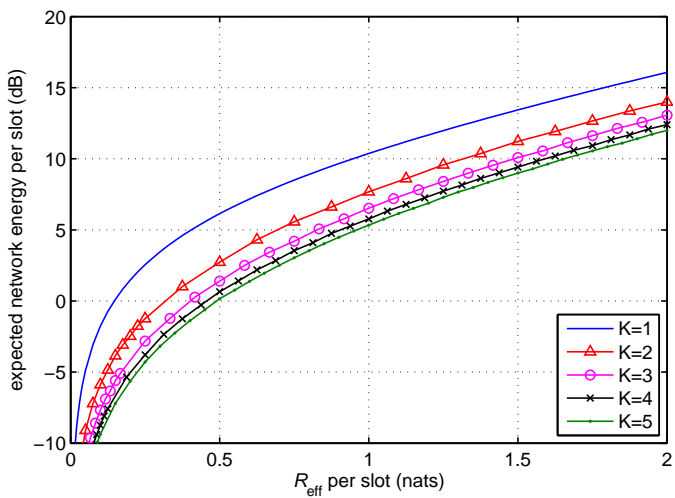}
\caption{Normalized minimum expected sum energy to deliver $\Reff$ nats per slot per channel use over $K$ slots. No relaying and $\snrtr=10^{-3}$.}
\label{fig:norelay_v0}
\end{figure}

We assume the SNRs are i.i.d. over different links and slots according to the truncated exponential PDF with average SNRs $\avesnrsd=\avesnrrd=\avesnrsr=1$; thus, any energy incurred may be taken as normalized with respect to the average SNR. We use the same SNR truncation threshold $\snrtr=10^{-3}$.
Fig.~\ref{fig:smalltrunc_v0} shows the NMESE $\bar{J}_K(\Reff)$ defined in \re{eqn:expectedJ_K} for various $K$.
The improvement is significant initially when $K$ is increased, but less so at larger $K$.
Our system setup may be considered a worst-case scenario, since typically $\relay$ is located somewhere between $\source$ and $\dest$ which implies that $\avesnrsd \ll \avesnrrd$ and $\avesnrsd \ll \avesnrsr$.
Nevertheless, the NMESE is still reduced when compared with Fig.~\ref{fig:norelay_v0} where no relay is used. For example at $\Reff=1$~nat, the reduction is more than $1$~dB and $0.5$~dB for $K=2$ and $K=5$, respectively.

\begin{figure}
\centering
\includegraphics[scale=\FigSize]{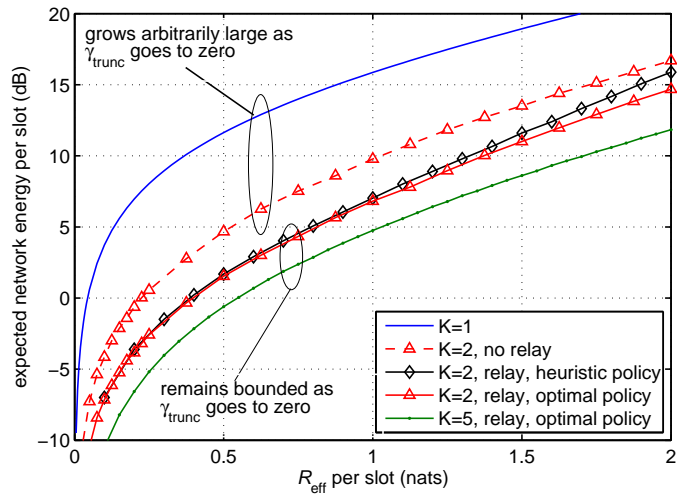}
\caption{A comparison of the normalized expected sum energy per slot per channel use over $K$ slots at different simulation conditions, assuming $\snrtr=10^{-10}$. A heuristic policy that is guaranteed to achieve a bounded expected sum energy is included for comparison.}
\label{fig:relay_110609_vvsmalltrunc_with_heu_paper}
\end{figure}

To investigate the effects of approaching a Rayleigh-fading channel, we reduce the SNR threshold to $\snrtr=10^{-10}$.
From Fig.~\ref{fig:relay_110609_vvsmalltrunc_with_heu_paper}, we see that the NMESE is almost unaffected if relaying is used, while the NMESE clearly increased if no relaying is used.
Fig.~\ref{fig:relay_110609_vvsmalltrunc_with_heu_paper} also highlights the significant progressive drop in NMESE as (i) $K=1$ is increased to $K=2$ without relay, (ii)  relay is added for $K=2$, and (iii)  $K=2$ is increased to $K=5$ with relay.
We also include a heuristic power allocation policy proposed in the proof of Theorem~\ref{thm:boundedenergy} for $K=2$, which is shown to achieve bounded sum energy even for $\snrtr\rightarrow 0$. In this scheme where $K=2$, $\source$ uses the minimum power to transmit in the first slot such that $\relay$ or $\dest$ decodes the message; if only $\relay$ decodes,  $\source$ or $\relay$ that has the stronger link to $\dest$ transmits such that $\dest$ can decode the message. From Fig.~\ref{fig:relay_110609_vvsmalltrunc_with_heu_paper}, the NMESE of the heuristic policy is close to the optimal scheme especially at small $\Reff$. To achieve a good complexity-performance tradeoff, other heuristic policies for different $K$ can also be devised and compared against the optimal schemes presented here.

\section{Conclusion}\label{sec:con}


In this paper, we consider the problem of energy minimization of a slotted relay system that guarantees data delivery within a deadline, assuming the availability of causal CSI.
We have obtained analytical solutions and provided conditions on which the expected sum energy is bounded.
Our results indicate that if full CSI is available, or if the channel is sufficiently rich in diversity, then the expected energy required is bounded. However, the expected energy becomes unbounded if only causal CSI is available (as is typical in practice) {\em and} if the channel is not sufficiently rich in diversity, such as Rayleigh and Rician channels that are representative of wireless channels. To make the expected energy bounded, advanced relaying with adaptive power allocation is necessary. Thus, we provide an alternative viewpoint of the advantage of relaying from the energy-efficient perspective.

An interesting future direction is on the analysis of the advantages of relaying in multi-user and multi-carrier systems from the energy-efficient perspective, compared to the case where no relaying is used such as in \cite{JoungHoTanSun12}.
As the exact analysis appears challenging, further asymptotic analysis is needed to yield insights.
The development of heuristic policies with provably good performance may also yield further insights for practical implementation.


\appendices

\section{Proof of Lemma~\ref{lem:a}}\label{append:proof:a}
We first prove that there is no loss of optimality if only either
$\source$  or $\relay$ transmits in each slot~$k$ in Phase~2, where $\widetilde{K}+1\leq k\leq K$ and $\widetilde{K}$ satisfies the constraint \re{p0:c1}.
Consider $\snrsd_k> \snrrd_k$.
Suppose the optimal power allocation $\mathbf{p}^{\star}_k=[\powro_k, \powso_k]$  satisfies $\powro_k>0$ and $\powso_k>0.$ A necessary condition for $\mathbf{p}^{\star}_k$ to be optimal is that it satisfies the constraint \re{p0:c2}.
By construction, we choose the non-negative powers $\mathbf{p}_k=(\powr_k,\pows_k)$ where $\powr_k=0$ and $\pows_k=\powso_k  + \powro_k \snrrd_k/\snrsd_k$.
Since $\powso_k \snrsd_k + \powro_k \snrrd_k=\pows_k \snrsd_k + \powr_k \snrrd_k$,  $\mathbf{p}_k$ satisfies the constraint \re{p0:c2}, i.e., this $\mathbf{p}_k$ is a feasible solution to Problem~$P0$.
Now, we write $\pows_k+\powr_k=\powso_k + \powro_k \snrrd_k/\snrsd_k < \powso_k  + \powro_k$ since $\snrsd_k> \snrrd_k$, and so $\mathbf{p}_k$ achieves a smaller sum energy than  $\mathbf{p}^{\star}_k$.
By contradiction,  $\mathbf{p}^{\star}_k$ is not optimal. Hence, it is optimal that only $\source$  or $\relay$ transmits.
The case of $\snrsd_k<\snrrd_k$ can be proved similarly.
%
Finally, consider $\snrsd_k=\snrrd_k$. Any non-negative $(\powr_k, \pows_k)$, such that $\powr_k+\pows_k$ is fixed, does not affect the constraint nor the objective function. Hence, no loss is incurred if  $\powro_k=0$ or $\powso_k=0$.
Now, given that only one node transmits, we complete the proof of Lemma~\ref{lem:a} that it is optimal for the {\em stronger} node to transmit. The proof is by contradiction: suppose only the {\em weaker node} with a strictly smaller SNR to the destination transmits, then a higher power needs to be allocated to satisfy the constraint \eqref{p0:c2}. This implies a smaller sum energy, i.e., a strictly suboptimal power allocation.


\section{Proof of Theorem~\ref{lem:2}}\label{proof:lem:2}
We use the Bellman's principle of optimality to solve Problem~$P1$ recursively, starting from slot $k=K$ until $k=1$ \cite{Bertsekas}.
It is useful to refer to the state diagram in Fig.~\ref{fig:transitions2}.
Given causal CSI $\state_k$ in slot $k$, $1\leq k\leq K$, let the sub-policy consisting of the power allocation of slots $k,\cdots,K,$ be $\pi_k(\state_k)\triangleq\{(p_k(\state_k),p_{k+1}(\state_{k+1}),\cdots,p_K(\state_K))
\geq \mathbf{0}, \forall \state_i\in {\mathcal{S}}_i, i=k+1,\cdots,K: \eqref{p1:c1} \text{ holds}\}.$

Consider the last slot $K$ in Phase $\phi\in\{1,2\}$ given causal CSI $\state$.
The minimum energy incurred for slot $K$ in Problem~$P1$, where the optimizing variable is $p_K(\state)=p_K$, is 
\be\label{app:eqn:JK}
J_{K,\phi}(\state) = \min_{p_K(\state)\in \pi_K(\state)}  \;  p_K.
\ee
The constraint \eqref{p1:c1} in $\pi_K(\state)$ is active only {\em for the last slot}, and so is equivalent to  $\Bd_{K+1} \leq 0$, i.e., the state must transit to $(K+1,\termstate)$.
From \re{eqn:amid},  we get $\Bd_{K} \leq \mi(p_K \snrsd_K)$ if $\phi=1$, and $\Bd_{K} \leq \mi(p_K \snrmax_K)$  if $\phi=2$.
Thus, $J_{K,\phi}(\state)$  can be obtained by \re{eqn:JK} for slot $K$.
Next, consider the next-to-last slot $k=K-1$ in Phase $\phi\in\{1,2\}$ given causal CSI $\state$.
We denote the next phase that the slot will transit into as $\phi'$ and the next causal CSI as $\state'$.
The minimum expected sum energy for slot $K-1$ and slot $K$  in Problem~$P1$, where the optimizing variables are  $p_{K-1}(\state)=p_{K-1}$ and the set $\{p_K(\state'), \forall \state'\}$, is then given by
\ben
&& J_{K-1,\phi}(\state) \nonumber \\
&=& \min_{(p_{K-1}(\state),p_K(\state'), \forall \state')\in\pi_{K-1}(\state)}   p_{K-1} + \mathbb{E}_{\phi', \state'}\left[ p_K | \phi, \state\right] \IEEEeqnarraynumspace\\
&=& \min_{p_{K-1}(\state)\geq 0} p_{K-1}  + \mathbb{E}_{\phi', \state'}\left[ \min_{p_K(\state')\in \pi_K(\state')} p_K \Big | \phi, \state\right] \\
&=& \min_{p_{K-1}(\state)\geq 0} p_{K-1}  + \mathbb{E}_{\phi', \state'}\left[ J_{K,\phi'}(\state')  | \phi, \state\right].
\een
In general, let $J_{k,\phi}(\state)$ denote the minimum expected sum energy for slot $k$ until slot $K$ in Problem~$P1$ given phase $\phi$ and causal CSI $\state$.
Repeating the proof similarly for decreasing $k=K-2,\cdots, 1$, we obtain \eqref{eqn:Jk} for all $k\in \mathcal{K}\backslash K$.
Finally, if $\phi=\termstate$,  the constraint \eqref{p1:c1} is satisfied and so it is optimal that $p^{\star}_k=0$,
hence we obtain \re{eqn:J_termstate}.
By backward recursion until $k=1$, we obtain $J_{1,1}(\state)$, which is the minimum expected sum energy for all slots since the system must start in Phase 1.
The optimal policy $\pi^*$ that achieves $J_{1,1}(\state)$ is clearly given by the power allocation that solves the backward recursive equations.

\section{Derivation of  \re{cor:K2trunc:1}-\re{eqn:2div}}
%
\label{append:K2trunc}
We obtain \re{cor:K2trunc:1}  from \re{eqn:JK} since the first slot equals the last slot and relaying is not possible, and we obtain  \re{cor:K2trunc:2} by applying Corollary~\ref{cor:K2}.
The derivation for \re{eqn:1div} follows by direct substitution and algebraic manipulations. The derivation for \re{eqn:2div} follows from the
%
 distribution of $\snrmax=\max\{ \snrsd, \snrrd\}$, which can be expressed as
$
f_{\snrmax}(x) = \frac{\partial }{\partial x} F_{\snrmax}(x) =  \frac{\partial}{\partial x}  \Pr\left( \snrsd\leq x,  \snrrd\leq  x \right)
= \frac{\partial}{\partial x}  F_{\snrsd}(x) F_{\snrrd}(x)
= f_{\snrsd}(x) F_{\snrrd}(x) + f_{\snrrd}(x) F_{\snrsd}(x)
$
for $x\geq  \snrtr$, where $F_X(\cdot)$ denotes the CDF of $X$. 
After some algebraic manipulations, we get $\Phi_2$ 
 as in \re{eqn:2div}.

\section{An Asymptotic Result}\label{append:lem:asym}

From \re{eqn:1div} and \re{eqn:2div}, we can express
\be
\Phi_1&=& \xi_1(\avesnrsd, \snrtr) \label{eqn:phi1_newexpress}\\
\Phi_2&=&\xi_2(\avesnrsd, \avesnrrd, \snrtr) \label{eqn:phi2_newexpress}
\ee
for $\avesnrsd, \avesnrrd >0$, where we denote
\be
\xi_1(x, \epsilon)&=&\exp\left(\epsilon/x\right) \Eone{\epsilon/x}/x  \label{eqn:xi1} \\
\xi_2(x, y, \epsilon)&=&\xi_1(x, \epsilon)+\xi_1(y, \epsilon)-\xi_1({\sf HM}(x,y), \epsilon) \label{eqn:xi2}
\ee
for $\epsilon\geq 0$ and $x> 0$, where $\Eone{x}=\int_x^{\infty}\exp(-t)/t \,\mathrm{d}t$ is the exponential integral and ${\sf HM}(x,y)=(1/x+1/y)^{-1}$ is the harmonic mean of $x$ and $y$.

\begin{lemma}\label{lem:asym}
For small $\epsilon$, we have the following asymptotic results:
\be\label{lem:asym1}
\xi_1(x, \epsilon) &=& -\frac{1}{x} \left( \gamma'   + \log\left(\frac{\epsilon}{x}\right)\right) + \mathcal{O}(\epsilon \log(\epsilon)) 
\\
\nonumber
\xi_2(x, y, \epsilon)
&=& \frac{1}{x} \log\left(1+\frac{x}{y}\right)+ \frac{1}{y} \log\left(1+\frac{y}{x}\right)
+ \mathcal{O}(\epsilon \log(\epsilon)).
\\
\label{lem:asym2}
\ee
where $\gamma'$ is the Euler's constant.
As $\snrtr\rightarrow 0$, we have
\be
\Phi_1 &\rightarrow& \infty \\
\Phi_2 &\rightarrow&  \frac{1}{\avesnrsd} \log\left(1+\frac{\avesnrsd}{\avesnrrd}\right)  +  \frac{1}{\avesnrrd} \log\left(1+\frac{\avesnrrd}{\avesnrsd}\right) \!< \!\infty.
\IEEEeqnarraynumspace
\ee
\end{lemma}
\begin{proof}
We use the asymptotic results that $\exp(\epsilon)=1+ \epsilon+ \mathcal{O}(\epsilon^2)$ and $\Eone{\epsilon}=-\gamma' -\log \epsilon + \mathcal{O}(\epsilon)$ \cite[Section~5.1.11]{Abramowitz:1965:HMF} for small $\epsilon$, to obtain \re{lem:asym1} and \re{lem:asym2} after some algebraic manipulations. Since $-\log(\epsilon) \rightarrow \infty$ and $\epsilon\log(\epsilon)\rightarrow 0$ as $\epsilon\rightarrow 0$, clearly $\Phi_1$ and $\Phi_2$ given in \re{eqn:phi1_newexpress} and \re{eqn:phi2_newexpress}, respectively, approach their respective bounds in Lemma~\ref{lem:asym}.
\end{proof}


%

\end{document}